\documentclass[11pt]{article}

\usepackage[linesnumbered,ruled]{algorithm2e}
\usepackage[margin=1in]{geometry}
\usepackage{hyperref}
\usepackage{amsfonts}
\usepackage{mathrsfs}
\usepackage{comment}
\usepackage{amsmath}
\usepackage{amssymb}
\usepackage{amsthm}
\usepackage{amscd}
\usepackage{graphicx}
\usepackage{indentfirst}
\usepackage[all]{xy}
\usepackage{titlesec}
\usepackage{enumerate}
\usepackage{bm}
\usepackage{enumitem}
\usepackage{color}
\usepackage{dsfont}
\usepackage{arydshln}
\newtheorem{theorem}{Theorem}[section]
\newtheorem{lemma}[theorem]{Lemma}
\newtheorem{proposition}[theorem]{Proposition}

\newtheorem{conjecture}[theorem]{Conjecture}
\newtheorem{definition}[theorem]{Definition}

\newtheorem{claim}[theorem]{Claim}
\newtheorem{example}[theorem]{Example}

\newtheorem{fact}[theorem]{Fact}

\newtheorem{observation}[theorem]{Observation}
\newcommand{\ma}{\mathcal}

\newcommand{\s}{\subseteq}

\newcommand{\fr}{\frac}
\newcommand{\lc}{\lceil}
\newcommand{\rc}{\rceil}
\newcommand{\lf}{\lfloor}
\newcommand{\rf}{\rfloor}

\newcommand{\x}{\vec}
\newcommand{\wt}{{\rm wt}}

\newcommand{\no}{\noindent}
\DeclareMathOperator{\rank}{rank}

\begin{document}
\title{Combinatorial list-decoding of Reed-Solomon codes beyond the Johnson radius}

\author{Chong Shangguan\footnote{Department of Electrical Engineering-Systems, Tel Aviv University, Tel Aviv 6997801, Israel. Email: theoreming@163.com.},
and Itzhak Tamo\footnote{Department of Electrical Engineering-Systems, Tel Aviv University, Tel Aviv 6997801, Israel. Email: zactamo@gmail.com.}
}

\maketitle

\begin{abstract}
List-decoding of Reed-Solomon (RS) codes beyond the so called Johnson radius has been one of the main open questions since the work of Guruswami and Sudan.
It is now known by the work of Rudra and Wootters,
using techniques from high dimensional probability, that over large enough alphabets most RS codes are indeed list-decodable beyond this radius.

In this paper we take a more combinatorial approach which allows us to determine the precise relation (up to the exact constant) between the decoding radius and the list size.
We prove a generalized Singleton bound for  a given list size, and  conjecture that the bound is tight for most RS codes over large enough finite fields. We also show that the conjecture holds true for list sizes $2 \text{ and }3$, and  as a by product  show that most RS codes with a rate of at  least  $1/9$   are list-decodable beyond the Johnson radius.
Lastly, we give the first explicit construction of such RS codes.
The main tools used in the proof are a new type of linear dependency between codewords of a  code that are contained in a small Hamming ball, and the notion of  cycle space from Graph Theory. Both of them have not been   used before in the context of list-decoding.
\end{abstract}

\newpage



\section{Introduction}
\noindent In this paper we consider the question  of   list-decodability   of error correcting codes, and in particular whether  Reed-Solomon (RS) codes are list-decodable  \emph{well-beyond} the so called Johnson radius. We give a strong affirmative answer to this major open question, as  detailed  below.

 Intuitively, a code is  list-decodable  if its codewords are well spread  in the space, i.e., there is no vector in the space whose Hamming ball of a certain radius around it contains `many' codewords of the code.  Formally, an  error-correcting code $\ma{C}$ of length $n$, over a finite alphabet $Q$ of size $q$,  is a subset of $Q^n$. For $r\in (0,1)$ and a positive integer $L$, we say that $\ma{C}$ is $(r,L)$ (combinatorial) list-decodable  if  for any vector $ \vec{y}\in Q^n$,  the Hamming ball $B_{r n}(\vec{y})$ of (relative) radius $r$ centered at $\vec{y}$ contains at most $L$ codewords of $\ma{C}$. In such a case, $r$ and $L$  are called the decoding radius and the list size, respectively.   Clearly, the problem of list-decoding  is   a generalization of the unique decoding problem typically considered in Coding Theory, where   given a  received word the decoder might output a list of possible codewords, instead of a unique one. Equivalently, the problem of list-decoding a code with a list of size $L=1$ is simply the unique decoding problem.

 Historically, the problem of list-decoding was first considered in the 50s independently by Elias \cite{elias} and  Wozencraft \cite{wozencraft}. Since then, it has attracted a lot of attention due to its various applications in Information Theory, (e.g.,  \cite{rudolf,Blinovski86,Blinovsky1997,Elias-survey91}) and  Complexity Theory \cite{Vadhan,Sudan}. The  question of list-decodability  is typically  divided into two questions  with two different flavours, a combinatorial question and an algorithmic question. The first studies the relation between the decoding radius $r$ and the list size $L$  for explicit or random codes over different alphabets \cite{Guru-Hastad-Koppa,Guru-Hastad-Sudan-Zuckerman,Wootters-2013,Cheraghchi-soda-2013}. Whereas, the algorithmic question asks whether the list of possible codewords can be efficiently computed given a received
word. In this paper we focus on the combinatorial question of list-decoding RS codes.

The \emph{list-decoding capacity theorem}
asserts that for $\epsilon >0$ there exist codes $\ma{C}\subseteq Q^n$ of rate
$R:=\frac{\log_q|\ma{C}|}{n}\geq   1-h_q(r) - \epsilon$  which are $(r, 1/\epsilon)$ list-decodable, where
$h_q(x) := x \log_q(q -1) - x \log_q (x) - (1 - x)\log_q
(1-x)$
is the $q$-ary entropy function. On the other hand, codes with rate
$R \geq 1 - h_q(r) + \epsilon$ are  \emph{only}   list-decodable with list size exponential in $n$.
Therefore, $1 - h_q(r)$ is termed as the {\it list-decoding capacity}. For a large alphabet size, say  $q\geq 2^{\Omega(1/\epsilon)}$,
$h_q(r)\leq r+\epsilon$, and then the theorem asserts that     for any $\epsilon>0$ there exist codes of rate $R\in (0,1)$ that are list-decodable from radius $1 - R - \epsilon$ and list size $L=O(1/\epsilon)$. Such a code is said  to achieve the list-decoding capacity.

The existence of capacity achieving codes follows from a standard random coding argument, which does not provide any explicit construction, nor any efficient decoding algorithm. Therefore, over the years there were many efforts into explicitly constructing such codes.   By now there are several examples of such  codes  \cite{Rudra,wang,kopparty}, among which \emph{folded RS codes} introduced by Guruswami and Rudra \cite{Rudra}  was the first example. However, the list-size guarantee in these codes is much larger than the $O(1/\epsilon)$ bound achieved by random codes, and is a large polynomial in the block length.
Another common property these codes share, is that they are all algebraic codes, and in particular they are all polynomial evaluation codes.

RS codes \cite{RS-codes} which were  introduced in the 60s, are by far the most fundamental examples of  evaluation codes, and therefore are the most  widely studied. Sudan  \cite{SUDAN1997} provided the first efficient list-decoding algorithm for RS codes,  which showed how to list-decode an  RS code with rate $R$  up to decoding radius of $1-\sqrt{2R}.$ Later, the decoding  radius was improved to $1-\sqrt{R}$ by the well-known  Guruswami-Sudan algorithm \cite{Guru-sudan-algo}. Surprisingly, the $1-\sqrt{R}$ decoding radius matches the Johnson bound  for RS codes which asserts that RS codes are  combinatorially list-decodable up to this radius, with a polynomial list size.
Since the work of Guruswami and Sudan, the question whether RS codes are list-decodable beyond the Johnson radius (either efficiently, or combinatorially) remained one of the central questions of this topic \cite{Guru-sudan-algo,Guruswami2007ListDA,Vadhan}.

Recently, this question was resolved by Rudra and Wootters \cite{Rudra-Wootters} who showed that in certain parameter regime,  with high probability, an   RS code with random evaluation points is list-decodable beyond the Johnson radius. However, the result is no means optimal, in the sense that the relation between the decoding radius, rate  and the list size is not optimal. Furthermore, the rate for which  it is possible to list-decode beyond the Johnson radius goes to zero as the alphabet size increases. We would like to mention that as far as we are aware,  this is the only positive result regarding list-decoding RS codes beyond the Johnson bound.

On the other hand, there were several results before the work of \cite{Rudra-Wootters} which hinted that  RS codes   might not be list-decodable beyond the Johnson bound.
It was shown by Guruswami and Rudra  \cite{Guruswami-rudra-limits-list-decoding} that RS codes are \emph{not} list-recoverable (which is a generalization of  list-decodability) beyond the Johnson radius. The next indication  was provided by  Ben-Sasson,  Kopparty  and   Radhakrishnan \cite{Ben-Sasson}, who showed that full length  RS codes (i.e., RS codes whose evaluation points are the entire finite field)  with rate $R$ are {\it not} list-decodable from radius $1-\sqrt[2-\alpha]{R}$ for any $\alpha>0$.
This result left open the possibility that one could carefully choose  the evaluation points to obtain RS codes which are combinatorially list-decodable well-beyond the Johnson bound. This   turned out be correct, as was shown  by the results  \cite{Rudra-Wootters} and the results presented in this paper.

Lastly, Cheng and Wan \cite{list-dec-discrete-log} showed a connection between efficient list-decoding of RS codes and the discrete log problem.  More precisely, they showed that the existence of an efficient  list-decoding algorithm  for RS codes of rate $R$ from radius of $1 -O(R)$, would imply the existence of an efficient  algorithm for the  discrete log problem. Notice that this result only sheds light on the algorithmic question of list-decoding.

\vspace{3mm}
{\bf Our contribution.} In this paper, we give an improved analysis of the list-decodability of random and explicit RS codes, and show that RS codes can be combinatorially list-decodable well-beyond the Johnson radius.  More precisely, our contribution is as follows.

\begin{itemize}
    \item {\bf A generalized Singleton bound.} The Singleton bound \cite{singletong-bound} is one of the fundamental bounds on the parameters of a code under  the
    problem of unique decoding. We give a simple generalization of it (which as far as we know is not known) to the problem of list-decoding.
    \item {\bf A Conjecture on the tightness of the generalized Singleton bound.} We phrase a new conjecture which asserts that over a large enough finite field, almost all RS codes attain the generalized Singleton bound. In order to prove the conjecture we develop a new machinery that enables us to capture information from many codewords of an RS code that are contained in a small Hamming ball. Currently this new approach enables us to prove the correctness of the conjecture for a list of sizes $L=2,3$.
    However, we believe that this machinery together with some new ideas can be used to prove the conjecture in its full generality, i.e., for any list size $L$.
    \item For list sizes $2$ and $3$ we show that almost all RS codes are list-decodable from radius $\frac{2}{3}(1-R)$ and $\frac{3}{4}(1-R)$, respectively. Note that the decoding radius for these codes are larger than the Johnson bound of $1-\sqrt{R}$ already for rates $R\geq 1/4$ and $1/9$, respectively.
    \item Although almost all RS codes (over a large enough finite field) are list-decodable beyond the Johnson radius, finding an explicit construction  would be ideal.   We give the first explicit construction for such codes.
\end{itemize}

\subsection{Definitions and Results}
\noindent A code $\ma{C}\subseteq Q^n$ is an arbitrary subset of vectors over an  alphabet of order $|Q|=q$, where its {\it rate} and {\it minimum (relative) Hamming distance}  are defined as $$R:=R(\ma{C})=\fr{\log_q|\ma{C}|}{n}, \delta:=\delta(\ma{C})=\min_{\x{x}\neq \x{y}\in \ma{C}}\frac{|\{i:x_i\neq y_i\}|}{n},$$
respectively. Vectors or codewords will be denoted with an arrow, i.e., by $\vec{y}$ and $\vec{c}$.
For  $r\in[0,1]$ and $L\in\mathds{Z}^+$, the code $\ma{C}$ is said to be {\it $(r,L)$ combinatorial list-decodable} if every ball of radius $r n$ in
$Q^n$  contains at most $L$ codewords of $\ma{C}$, i.e., for any $ \vec{y}\in Q^n$
$$ |B_{r n}(\vec{y})\cap\ma{C}|\le L,$$
where $B_{rn}(\vec{y})$ is the ball of radius $r$ centered at $\vec{y}$. In such a case, $L$ and $r$ will be referred as the {\it list size} and  the {\it decoding radius}, respectively.
For short, in the sequel we will omit the word combinatorial whenever we refer to list-decodable codes.

The main goal  of this  paper is to gain more understanding on the question of given list size $L$, what the largest possible $r$ is for which a code is $(r, L)$ list-decodable.
For  $L=1$ the value of $r$ 
is completely determined by  the minimum distance of the code. Indeed, the minimum distance of the code is at least $\delta$ if and only if the code is  $(\frac{\delta}{2}-\frac{1}{2n},1)$ list-decodable.
For  $L>1$ the problem is much more challenging, and it is probably impossible to  answer  it  based solely on the standard parameters of a code.  In spite of that,  we will  provide in this paper a new upper bound of $r$ as a function of $n,k, L$, and a matching lower bound for  $L=2,3$, as detailed next.

A natural  first attempt towards shedding more light on the problem  would be by generalizing the known bounds for unique decoding, i.e., $L=1$, to arbitrary $L$.
The well-known  Singleton bound, where by the above notation states the following.
\begin{theorem}
\label{singleton}
If $\ma{C}\subseteq Q^n$ is   $(r,1)$ list-decodable then   $|\ma{C}|\leq q^{n(1-2r)}.$
\end{theorem}
Our first result is a generalization of this bound to an arbitrary list size $L$. Although the result is  simple, as far as we know, it is new and  has never appeared  before in  the literature.

\begin{theorem}[Generalized Singleton bound]\label{singleton-type}
If $\ma{C}\subseteq Q^n$ is an $(r,L)$ list-decodable code, then
\begin{equation}
\label{mish}
|\ma{C}|\le Lq^{n-\lf\frac{(L+1)r n}{L}\rf}.
\end{equation}
Moreover, if $\ma{C}$ is a linear code over $\mathbb{F}_q$ with   $q>L$, then
\begin{equation}\label{mish2}
|\ma{C}|\le q^{n-\lf\frac{(L+1)r n}{L}\rf}.
\end{equation}
\end{theorem}

Note for a prime power $q$,  $\mathbb{F}_q$ denotes the finite field of $q$ elements. An {\it $[n,k]$-linear code} is a subspace of dimension $k$ of $\mathbb{F}_q^n$.
Clearly, one can verify that Theorem \ref{singleton-type} is indeed a generalization of the Singleton bound, since   Theorem  \ref{singleton} is recovered by setting  $L=1$.
For ease of presentation, in what follows we always assume that the parameters of the code satisfy that $\fr{r n}{L}$ is an integer, and therefore the floor operation in Theorem \ref{singleton-type} can be removed.
%
%

Linear codes that attain the Singleton bound, i.e., attain \eqref{mish} and \eqref{mish2}  with equality for $L=1$ are called {\it MDS codes}, and the famous family of RS codes defined below are known to be such codes.
\begin{definition}
Given  $n$ distinct elements $\alpha_i\in \mathbb{F}_q$, the $[n,k]$-RS defined by the evaluation vector $\vec{\alpha}=(\alpha_1,...,\alpha_n)$ is the set of vectors
$$\{(f(\alpha_1),\ldots,f(\alpha_n)):f\in \mathbb{F}_q[x],  \deg(f) <k \}.$$
\end{definition}
It is natural to ask whether RS codes also attain the bounds with equality for $L\geq 2$, and in particular attain \eqref{mish2} for fixed $L$ and large $n,q$. The following conjecture exactly states that this is indeed the case,  but before presenting it, we  will need to  rephrase   the second claim of Theorem \ref{singleton-type} in the following way:  If  $q>L$ then an $(r,L)$ list-decodable $[n,k]$-linear code must satisfy
\begin{equation}\label{n,k,r,L}
  \begin{aligned}
    r\le \fr{L(n-k)}{(L+1)n}=\frac{L}{L+1}(1-R),
  \end{aligned}
\end{equation}
and codes that attain \eqref{n,k,r,L} with equality will be called {\it optimal $L$ list-decodable}.
\begin{conjecture}\label{conjecture-0}
  Let $L\ge 2$ and $n,k$ be fixed integers,   then there is a constant $c:=c(n,k,L)$ such that the following statement holds.
  If $q>c$ is a prime power, then all but at most $cq^{n-1}$  vectors of $\mathbb{F}_q^n$ define an  $[n,k]$-RS code which is $(\fr{L}{L+1}(1-R),L)$ list-decodable.
\end{conjecture}
From a probabilistic point of view, Conjecture \ref{conjecture-0} implies that for given $n,k,L$ and sufficiently large $q>c(n,k,L)$, a random RS code defined by uniformly picking  its  evaluation vector  from  the set of all evaluation vectors, is $(\fr{L}{L+1}(1-R),L)$ list-decodable  with probability at least $1-\ma{O}_{n,k,L}(\fr{1}{q})$.
   A similar phenomena was observed in \cite{Rudra-Wootters} (see the discussion below).
Unfortunately, we are not able to prove Conjecture \ref{conjecture-0} in its full generality, but only for $L=2,3$, as given below in Theorems \ref{2 list-optimal} and \ref{3 list-optimal}. However, we believe that our approach  can be used to prove the conjecture completely. As a step towards this, we pose another conjecture (see Conjecture \ref{conjecture} below) whose correctness would imply the correctness of  Conjecture \ref{conjecture-0}.
If Conjecture \ref{conjecture-0} holds true, then so does  the following conjecture which gives the exact order (including the exact constant) of the list size as a function of the rate $R$ and distance  $\epsilon$ from capacity.
\begin{conjecture}
 For any $R,\epsilon>0$, there exist RS codes with rate $R$ over a large enough finite field, that are list-decodable from radius $1-R-\epsilon$ and list size at most
 $\frac{1-R-\epsilon}{\epsilon}.$
\end{conjecture}
\subsection{Optimal 2 and 3 list-decodable RS codes}\label{subsection-2-LD}
\smallskip
\noindent{\bf Optimal $2$ list-decodable RS codes:}
 The following Theorem shows that Conjecture \ref{conjecture-0} holds true for $L=2$, i.e., for sufficiently large $q$, almost all $[n,k]$-RS codes in $\mathbb{F}_q^n$ are $(\fr{2}{3}(1-R),2)$ list-decodable. Unfortunately, the proof requires the field size to grow  exponentially with $n$. It is not clear whether this is an artifact of the proof and in fact a much smaller field size would suffice. In other words, is it similar to the  existence proof of MDS codes, which requires  an  exponentially large (in the code length) field size, whereas it is known that only a linear field size   suffices.
\begin{theorem}[An existence proof]\label{2 list-optimal}
  Let $n>k$ be integers and set $c=2^{3n}k^2+n^2$, then all but at most $cq^{n-1}$ vectors of $\mathbb{F}_q^n$ define
a $(\fr{2}{3}(1-R),2)$ list-decodable  $[n,k]$-RS code.
\end{theorem}
Clearly, Theorem \ref{2 list-optimal} is nontrivial only if $q>c$. In such a case,
it  shows that most  RS codes are optimal $2$ list-decodable, however it does not construct any code  explicitly.  Fortunately, Theorem \ref{2 list-explicit} below provides such a construction over an even  larger field size ($q$ is   double-exponential  in $n$ for $k=\Theta(n)$).

\begin{theorem}[Explicit construction]\label{2 list-explicit}
Set $q=2^{k^n}$ and let $\vec{\alpha}=(\alpha_1,\ldots,\alpha_n)\in \mathbb{F}_q^n$ be a vector such that for $1\leq i\leq n$
\begin{equation}
\label{Eq-0}
\mathbb{F}_{2^{k^{i-1}}}(\alpha_i)=\mathbb{F}_{2^{k^i}},
\end{equation}
 i.e., $\alpha_i$ generates a degree $k$ extension over the field $\mathbb{F}_{2^{k^{i-1}}}$.  Then, the $[n,k]$-RS code defined by $\vec{\alpha}$ is $(\fr{2}{3}(1-R),2)$ list-decodable.
%
%
%
%
\end{theorem}

Typically the list-decodability of codes is compared to the Johnson bound, which shows that \emph{any} code is list-decodable up to a certain radius (called the Johnson radius) with a polynomial size list, provided that the alphabet size is also polynomial in the length of the code. The precise statement is as follows.
\begin{theorem}[Johnson bound, see, e.g., \cite{guruswami2019essential} Theorem 7.3.3]
Any code with minimum Hamming distance $\delta$ is $(1-\sqrt{1-\delta},qn^2\delta)$ list-decodable for any alphabet $q$.
\end{theorem}
In particular, an $[n,k]$-RS code is $\big(1-\sqrt{\fr{k-1}{n}},qn(n-k+1)\big)$ list-decodable, i.e., for $k,n\rightarrow \infty$ and fixed rate $k/n=R$ the Johnson bound claims that RS codes are combinatorially list-decodable up to radius of $1-\sqrt{R}$ and  a polynomial list size, provided that $q$ is polynomial in $n$.
Moreover, the celebrated Guruswami-Sudan algorithm \cite{gurus} is  an  efficient algorithm for the list-decoding problem of RS codes up to  the Johnson radius.
The question of whether RS codes are list-decodable beyond the Johnson radius was left open since the work of \cite{gurus}. In fact, this was one of the major open questions of this topic  (see, e.g., \cite{guruswami2019essential} Open Question 15.2.1). Note that it was not even known that whether these codes are combinatorially list-decodable beyond this radius, without even considering the algorithmic question. Recently, this question was finally settled  by Rudra and Wootters \cite{Rudra-Wootters} who showed that
with high probability an  RS code  with random evaluation points from the field $\mathbb{F}_q$ and rate
$$\Omega(\frac{\epsilon}{\log(q) \log^5(1/\epsilon)}),$$
is list-decodable up to radius of  $1-\epsilon$ with list size $O(1/\epsilon)$, and  this beats the Johnson bound whenever $\epsilon\leq  \tilde{O}(1/ \log(q))$.
Note that as  $q$ increases, the rate  for which it is possible to beat the Johnson bound goes to zero. Furthermore, this result only gives the orders of magnitude  of these parameters, and not their exact values.

On the other hand, our result reveals the \emph{exact} behaviour of the possible decoding radius as a function of the list size.  In particular, since $\fr{2}{3}(1-R)\geq 1-\sqrt{R}$ for $R\in [\fr{1}{4},1]$, it follows from  Theorem \ref{2 list-optimal} that for  sufficiently large $q$,  almost all $[n,k]$-RS codes with rate  $R>\fr{1}{4}$ are combinatorially  list-decodable beyond the Johnson radius, and  moreover, the list size is guaranteed to be at most  two.

Theorem \ref{2 list-explicit} provides the \emph{first} explicit construction of RS codes that are list-decodable beyond the Johnson radius.
Clearly, the main open question for these codes is the following algorithmic question: Are these codes efficiently list-decodable (see more on this  in the list of open questions in Section \ref{conclustions})?

Notice that the Johnson bound does not provide any interesting information when the alphabet size is extremely large compared to the length of the code, e.g., $q$  is exponential in $n$, as in the case of Theorem \ref{2 list-optimal}. At a first glance, the fact that the bound on the list size in the Johnson bound increases with  $q$ might seem reasonable, since the size of the code increases with  $q$ too. Consequently, there might be Hamming balls with radius equal to the Johnson radius that  contain many codewords. However, our above results show that this is not the case for almost all RS codes.

\smallskip
\noindent{\bf Optimal 3 list-decodable RS codes:} Our last result shows that Conjecture \ref{conjecture-0} holds true also for $L=3$, i.e., over sufficiently large alphabet, almost all RS codes are optimal  $3$ list-decodable.
\begin{theorem}\label{3 list-optimal}
Let $n>k$ be positive integers,  then there is a constant $c_3:=c_3(n,k)$
  such that all but at most $c_3(n,k)q^{n-1}$ vectors of $\mathbb{F}_q^n$ define a $(\fr{3}{4}(1-R),3)$ list-decodable $[n,k]$-RS code.
\end{theorem}
The proof of Theorem \ref{3 list-optimal} is very similar to that of Theorem \ref{2 list-optimal}, although  it is much more technical and involved.
Continuing the above comparison to the Johnson bound, Theorem \ref{3 list-optimal} implies that almost all RS codes with rate $R>\fr{1}{9}$ and sufficiently large $q$ are  list-decodable beyond the Johnson radius.
Lastly,  similar to the explicit construction in Theorem \ref{2 list-explicit}, an explicit construction of a $(\fr{3}{4}(1-R),3)$ list-decodable $[n,k]$-RS code over a field of size $q=2^{(3k)^n}$ can be given. We omit the details.

 In the course of proving the results in the paper we develop a machinery that might be useful to generalize them to any list size $L$. In particular, we connect the problem of list-decoding to the cycle space of the complete graph (which is a new connection) and introduce the notion of  intersection matrices (see Section \ref{preparation} below).
The key ideas of our approach are given in the next section.

\subsection{Outline of the paper}
\noindent The rest of the paper is organized as follows. In Section \ref{keyideas} we provide  an overview of the key ideas and some needed  notations.
The proof of Theorem \ref{singleton-type} is presented in Section \ref{singleton-proof}.
In Section \ref{proof-of-mian-theorem-1}, we give the definition of 3-wise intersection matrices and use it to prove Theorems \ref{2 list-optimal} and \ref{2 list-explicit}.
In Section \ref{preparation} we present the notion of cycle spaces and define $t$-wise intersection matrices for $t\ge 3$, and use them to prove Conjecture \ref{conjecture-0} assuming  Conjecture \ref{conjecture} holds true.
In Section \ref{proof-of-mian-theorem-2} we present the proof of Theorem \ref{3 list-optimal}.
We conclude the paper in Section \ref{conclustions} with some concluding remarks and open questions.

\section{An overview of the the key ideas and notations}
\label{keyideas}

\no For fixed positive integers $n>k$ and $L$, an evaluation vector $\x{\alpha}\in\mathbb{F}_q^n$ is said to be {\it bad} if the $[n,k]$-RS code defined by $\x{\alpha}$ is not $(\fr{L}{L+1}(1-R),L)$ list-decodable, i.e., it does not achieve (\ref{n,k,r,L}) with equality.

Our  proof strategy would be to show that there exist many evaluation vectors that are not bad. First, note that that the  number of distinct evaluation vectors is  of order  $\fr{q!}{(q-n)!}=\Theta_n(q^n)$.  Hence, if we manage to show that the number of bad evaluation vectors is at most $\ma{O}_{n,k}(q^{n-1})$, then for sufficiently large $q$, all but at most an $\ma{O}_{n,k}(\fr{1}{q})$-fraction (which tends to zero as $q$ tends to infinity) of evaluation vectors are not bad. In other words, they define an $(\fr{L}{L+1}(1-R),L)$ list-decodable $[n,k]$-RS code, as needed.

Towards this end, we  bound the number of bad evaluation vectors as follows. We show that there exists a family $\mathbb{F}_S\s \mathbb{F}_q[x_1,\ldots,x_n]$ of \emph{nonzero} polynomials such that

\begin{itemize}
  \item [(i)] The degree of every polynomial in $\mathbb{F}_S$ is bounded from  above by a function  $d(k)$;

  \item [(ii)] The size of $\mathbb{F}_S$ is bounded from  above by a function $\mu(n,k)$;

  \item [(iii)] Every bad evaluation vector is a zero of some polynomial in $\mathbb{F}_S$.
\end{itemize}

\noindent Consequently,   all  bad evaluation vectors are contained in the union of the zero sets of the polynomials in $\mathbb{F}_S$.
On the other hand, the  number of zeros of a nonzero polynomial is upper bounded in the following standard lemma.

\begin{lemma}[DeMillo-Lipton-Schwartz-Zippel lemma, see, e.g., \cite{jukna2011extremal} Lemma 16.3]\label{zeros}
A nonzero polynomial $f\in\mathbb{F}_q[x_1,\ldots,x_n]$ of degree at most $d$ has at most $dq^{n-1}$ zeros  in $\mathbb{F}_q^n$.
\end{lemma}

\no Therefore, by the union bound the number of bad evaluation vectors is at most
$$|\mathbb{F}_S|\cdot\big(d(k)q^{n-1}\big)\le \mu(n,k) d(k)q^{n-1}=\ma{O}_{n,k}(q^{n-1}),$$
as desired.

The  family of nonzero polynomials $\mathbb{F}_S$ is the main ingredient in the proof.  Its construction constitutes the main technical contribution, and  is based on the cycle space (a notion from Graph Theory, see Section \ref{cycle-space} below) of the complete graph $K_{L+1}$ and  a class of auxiliary matrices which we call $(L+1)$-wise intersection matrices (defined in  Section \ref{int-matrix-def-gen} below).

Loosely speaking, if an evaluation vector $\x{\alpha}$ is bad, then by definition there exists a vector $\x{y}\in\mathbb{F}_q^n$ and $L+1$ distinct codewords $\x{c}_1,\ldots,\x{c}_{L+1}\in\ma{C}$ in the code defined by it,  such that $\{\x{c}_1,\ldots,\x{c}_{L+1}\}\s B_{\frac{L(n-k)}{(L+1)}}(\x{y})$.
Since $\x{c}_i$ are `close' in the Hamming distance to $\x{y}$, by the triangle inequality they are close to each other, i.e., any two codewords must agree on `many' coordinates.
$(L+1)$-wise intersection matrices, which are  variable matrices in the variables  $x_1,\ldots,x_n$, capture  this information, i.e., the large number of  agreements among the $L+1$ codewords $\x{c}_i$.
Assuming the evaluation vector  $\x{\alpha}$ is bad, we show that there exists a square submatrix of an $(L+1)$-wise intersection matrix whose determinant is a \emph{nonzero} polynomial in $\mathbb{F}_q[x_1,\ldots,x_n]$, which vanishes at $\x{\alpha}$. The determinants of these submatrices  form the family of polynomials $\mathbb{F}_S$.

Unfortunately, the step of showing the existence of such a submatrix is nontrivial, and we are only able to prove it for $L=2,3$. The proof for $L=2$ uses the notion of   $3$-wise intersection matrices, and  is relatively short.
Whereas the proof for $L=3$ uses similar arguments together with some new ideas and notions, but is much more involved.

\subsection{Notations}\label{prel}

\noindent We will use the standard Bachmann-Landau notations $\Omega(\cdot),\Theta(\cdot),\ma{O}(\cdot)$ and $o(\cdot)$. For a positive integer $n$, let $[n]:=\{1,\ldots,n\}$, and for a subset $V\subseteq[n]$ and a positive integer $t$, let $\binom{V}{t}$ (resp. $\binom{V}{\le t}$)  be the family of all subsets of $V$ of size $t$ (resp. at most $t$).  For a vector $\x{c}=(c_1,\ldots,c_n)\in Q^n$ over an alphabet $Q$ and a subset $I\s[n]$, let $(\x{c})_{I}=(c_i:i\in I)\in Q^{|I|}$ be restriction of $\x{c}$ to  the coordinates with indices in $I$.

Throughout the  paper, $n$ and $k$ are the length and dimension of the code, whereas $L$ is the list size.  $q$ is  a prime power such that $q\ge \max\{n,L+1\}$, and $\ma{C}\s\mathbb{F}_q^n$ is   an $[n,k]$-RS code defined by  $n$ distinct evaluation points $\alpha_1,\ldots,\alpha_n$.
The vector $\vec{\alpha}=(\alpha_1,\ldots,\alpha_n)\in\mathbb{F}_q^n$ is called the {\it evaluation vector} of $\ma{C}$, and we say that $\ma{C}$ is {\it defined} by $\vec{\alpha}$.

Given  $n$  variables or field elements $x_1,\ldots,x_n$ and an integer $s\ge 1$ (note that we will be mostly interested in the case $s=k$), define the  $n\times s$ Vandermonde matrix 
  $$V_s(x_1,\ldots,x_n)=\left(
    \begin{array}{cccc}
      1  & x_{1} & \cdots & x_{1}^{s-1} \\
       &  & \ddots &  \\
      1  & x_{n} & \cdots & x_{n}^{s-1} \\
    \end{array}
  \right),$$
and for  $I\s[n]$, let  $V_s(x_i: i\in I)$  be the restriction of $V_s(x_1,\ldots,x_n)$ to the rows with indices in $I$. When the $x_i$'s are understood from the context, we will use also the abbreviation
 $V_s(I)$. Denote by  $\mathbb{F}_q^{<s}[x]$   the subspace of polynomials of degree less than $s$ in the variable $x$.
By abuse of notation, we view each polynomial  $f=\sum_{i=0}^{s-1}c_ix^i\in\mathbb{F}_q^{<s}[x]$ also as a length $s$ vector $f=(c_0,\ldots,c_{s-1})\in\mathbb{F}_q^s$.
%
Lastly, for an integer $t\ge 1$ and  a  list of  vectors $\vec{y}_1,\ldots,\vec{y}_t\in\mathbb{F}_q^n$, let $I(\vec{y}_1,\ldots,\vec{y}_t)\s[n]$ be the collection of indices for which $\vec{y}_1,\ldots,\vec{y}_t$ agree on, i.e.,   $i\in I(\vec{y}_1,\ldots,\vec{y}_t)$ if and only if the $i$-th coordinate of the vectors $\x{y}_j, 1\leq j \leq t$  is identical.

The following easy observation will be useful. 
\begin{observation}\label{matrix-example-0}
With the above notation, for a polynomial $f\in\mathbb{F}_q^{<s}[x]$ (recall that by abuse of notation $f$ is also viewed as a vector in $\mathbb{F}_q^s$) and a vector $\x{\alpha}=(\alpha_1,\ldots,\alpha_n)\in\mathbb{F}_q^n$,
$$V_s(\alpha_i:i\in I)\cdot f^T=0$$ for any $\{\alpha_i:i\in I\}\s Z(f)$, where $Z(f)\s\mathbb{F}_q$ is the {\it zero set} of the polynomial $f$.
\end{observation}

\section{A Singleton-type upper bound}\label{singleton-proof}

\noindent In this section we give the proof of the Singleton-type bound.

\begin{proof}[\textbf{Proof of Theorem \ref{singleton-type}.}]
It is enough to prove the first statement of the theorem, as the second statement follows easily from it.
  Let $a:=\lf\frac{(L+1)rn}{L}\rf= rn+ \lf\frac{rn}{L}\rf$ (assuming $rn$ is an integer), and assume to the contrary that $|\ma{C}|> Lq^{n-a}$. We will show that  there exists a vector $\vec{y}\in Q^n$ such that $|{\rm B}_{rn}(\vec{y})\cap \ma{C}|\ge L+1$, thereby  violating the $(r,L)$ list-decodability of the code.

  By the pigeonhole principle there exist at least $L+1$ distinct codewords $\vec{c}_1,\ldots,\vec{c}_{L+1}$ of $\ma{C}$ which  agree on their  first $n-a$ coordinates, i.e.,
  $(\vec{c}_1)_{[n-a]}=(\vec{c}_2)_{[n-a]}=\cdots=(\vec{c}_{L+1})_{[n-a]}$.
%
%
  Partition arbitrarily the set of the last $a$ coordinates $\{n-a+1,\ldots,n\}$  as evenly as possible to $L+1$ subsets $I_1,\ldots,I_{L+1}$ each of size  $\lf\fr{a}{L+1}\rf$ or $\lc\fr{a}{L+1}\rc$, and define the vector  $\vec{y}\in[q]^n$ to be
  $$(\vec{y})_{[n-a]}=(\vec{c}_1)_{[n-a]},\text{ and } (\vec{y})_{I_i}=(\vec{c}_i)_{I_i} \text{ for $1\le i\le L+1$.}$$
  Clearly, $\vec{y}$ is well defined as $[n-a], I_1,\ldots,I_{L+1}$ form a partition of $[n]$.
  Moreover,
  for  $1\le i\le L+1$ the Hamming distance between $\x{c}_i$ and $\x{y}$ is at most
  \begin{align*}
  a-|I_i|\le a- \lf\fr{a}{L+1}\rf &= rn+ \lf\frac{rn}{L}\rf- \lf \frac{rn}{L+1}+\lf \frac{rn}{L}\rf \frac{1}{L+1} \rf\\
      & \leq rn+ \lf\frac{rn}{L}\rf- \lf \frac{rn}{L+1}+ \frac{rn}{L} \frac{1}{L+1} \rf=rn.
  \end{align*}  Hence,  $\{\vec{c}_1,\ldots,\vec{c}_{L+1}\}\s {\rm B}_{rn}(\vec{y})$, and we arrive at  a contradiction.
\end{proof}

\section{Optimal 2 list-decodable RS codes}\label{proof-of-mian-theorem-1}

\no Following the discussion in Section \ref{keyideas}, we prove in this section Theorems \ref{2 list-optimal} and \ref{2 list-explicit}.
The ideas presented in this section will be further developed in Sections \ref{preparation} and  \ref{proof-of-mian-theorem-2} in order to  prove Theorem \ref{3 list-optimal}.

The high level idea behind the proof 
is as follows. We will show that if an $[n,k]$-RS code defined by an evaluation vector $\vec{\alpha} \in \mathbb{F}_q^n$ is \emph{not} $(\frac{2}{3}(1-R),2)$ list-decodable, then $\vec{\alpha}$ must annihilate  certain nonzero polynomial of degree at most $d:=d(k)$. We further show that the number of such polynomials   is bounded from above by some function of $n, k$. Hence, by applying  Lemma \ref{zeros} and the union bound, the number of such vectors $\vec{\alpha}$ is at most $c(n,k)q^{n-1}$, where $c(n,k)$ depends only on $n$ and $k$. Put differently, if the underlying field is of   order $q>c(n,k)$, then  there exists an evaluation vector $\vec{\alpha}\in \mathbb{F}_q^n$ which defines a $(\frac{2}{3}(1-R),2)$ list-decodable $[n,k]$-RS code.

\subsection{A necessary condition for RS codes not to be   $(\fr{2}{3}(1-R),2)$ list-decodable}

\noindent In this subsection we prove  Lemma \ref{necessary} which provides a necessary condition for an $[n,k]$-RS code \emph{not} to be  $(\fr{2}{3}(1-R),2)$ list-decodable. Then, by showing that there exist RS codes which do not  satisfy such a condition we conclude that there exist $(\fr{2}{3}(1-R),2)$ list-decodable RS codes.

 The following definition of  $3$-wise intersection matrix is the main algebraic object in the proof of Theorems \ref{2 list-optimal} and \ref{2 list-explicit}. Based on the cycle space of the complete graph $K_t$ for arbitrary $t\ge 3$, We will generalize this definition in the next section    to $t$-wise intersection matrix.

\begin{definition}[3-wise intersection matrix]\label{Def-matrix-1}
  For an integer $s\ge 1$ and three subsets $I_1,I_2,I_3\s [n]$, the 3-wise intersection matrix $M_{s, (I_1,I_2,I_3)}$ is an $(s+\sum_{1\le i<j\le 3}|I_i\cap I_j|)\times 3s$ variable matrix   (in the variables $x_1,\ldots,x_n$) defined as
$$M_{s, (I_1,I_2,I_3)}:=\left(
    \begin{array}{c|c|c}
      \mathcal{I}_s & \mathcal{I}_s & \mathcal{I}_s \\\hline
      V_{s}(I_1\cap I_2) &  &  \\\hline
       & V_{s}(I_1\cap I_3) &  \\\hline
       &  & V_{s}(I_2\cap I_3) \\
    \end{array}
  \right),$$
where $\mathcal{I}_s$ is the identity matrix of order $s$.
\end{definition}

For a vector $\vec{\alpha}=(\alpha_1,\ldots,\alpha_n)\in \mathbb{F}_q^n$, let $M_{s, (I_1,I_2,I_3)}(\vec{\alpha})$ be the matrix over $\mathbb{F}_q$ where each variable $x_i$ is assigned the value $\alpha_i$.
It is easy to verify that $M_{s, (I_1,I_2,I_3)}(\vec{\alpha})$ has a
nonzero kernel, i.e., there exist $f_1,f_2,f_3 \in \mathbb{F}_q^s$, not all zero, such that
$$M_{s, (I_1,I_2,I_3)}(\vec{\alpha})\cdot (f_1,f_2,f_3)^T=0,$$
if and only if the  polynomials $f_1,f_2,f_3 \in \mathbb{F}_q^{<s}[x]$
satisfy

\begin{itemize}
\item $f_1+f_2+f_3=0$,
\item$ \{\alpha_i:i\in I_1\cap I_2\}\s Z(f_1), \{\alpha_i:i\in I_1\cap I_3\}\s Z(f_2) \text{ and } \{\alpha_i:i\in I_2\cap I_3\}\s Z(f_3)$.
\end{itemize}


%
%

We will be interested in $3$-wise intersection matrices defined by triples $(I_1,I_2,I_3)$ that satisfy  the following special properties. 
Let $S$ be the set of triples $(I_1,I_2,I_3)$ of subsets of $[n]$ that satisfy
\begin{itemize}
\item $I_1\cap I_2\cap I_3=\emptyset$,
\item $\sum_{1\le i<j\le 3}|I_i\cap I_j|$ is even,
\item for distinct $i,j\in [3]$,
$$|I_i\cap I_j|\leq \frac{1}{2}\sum_{1\le i<j\le 3}|I_i\cap I_j|\leq k.$$

\end{itemize}

The next lemma shows that for triples $(I_1,I_2,I_3)\in S$ and $s=\frac{1}{2}\sum_{1\le i<j\le 3}|I_i\cap I_j|$, $\det(M_{s, (I_1,I_2,I_3)})$ is a nonzero polynomial in $\mathbb{F}_q[x_1,\ldots,x_n]$.

\begin{lemma}\label{determinant-1}
  Let $(I_1,I_2,I_3)\in S$ and set  $A=I_1\cap I_2, B=I_2\cap I_3, C=I_3\cap I_1,
  \text{ and } s=\frac{1}{2}(|A|+|B|+|C|)$.
  Then, $\det(M_{s, (I_1,I_2,I_3)})$ is a nonzero polynomial in $\mathbb{F}_q[x_1,\ldots,x_n]$ with degree  $s(s-1)$, where each variable is of degree at most $s-1$. In particular,  the    monomial
\begin{equation}
\label{monomial}
\prod_{a\in A}x_a^{t_a}\prod_{b\in B}x_b^{t_b}\prod_{c\in C}x_c^{t_c}
\end{equation}  appears in the
polynomial $\det(M_{s, (I_1,I_2,I_3)})$
with a nonzero coefficient, if and only if
\begin{enumerate}
    \item $\{t_a:a\in A\}\cup \{t_b:b\in B\} \cup \{t_c:c\in C\}$ viewed as a the  multiset of size $2s$, contains each integer $0\le i\le s-1$ exactly
twice.
    \item the $t_a$'s for $a\in A$ are all distinct, and similarly, the $t_b$'s and $t_c$'s, i.e., $\{t_a:a\in A\}$, $\{t_b:b\in B\}$ and  $\{t_c:c\in C\}$ are sets and not multisets.
\end{enumerate}
 Furthermore, the nonzero coefficient of such a monomial is either $1$ or $-1$.
\end{lemma}

\begin{proof}
Observe that by definition $M_{s,(I_1,I_2,I_3)}$ is a square matrix of order $3s$, so $\det(M_{s, (I_1,I_2,I_3)})$ is well defined.
Assume that the union of the sets $\{t_a:a\in A\},\{t_b:b\in B\}, \{t_c:c\in C\}$ (viewed as a multiset) contains   each integer $0\le i\le s-1$ exactly
twice. By the  definition of $S$, the sets  $A,B,C$ are pairwise disjoint and the size of each set  is at most $s$, hence there is exactly one way to get the monomial \eqref{monomial} in the determinant expansion of $M_{s, (I_1,I_2,I_3)}$. Moreover,  it is easy to verify that  after removing the $2s$ rows and columns of
  $M_{s, (I_1,I_2,I_3)}$ which correspond to the monomial \eqref{monomial} in the determinant expansion, 
  the resulting submatrix is the identity matrix of order $s$ with possible permuted columns, hence the monomial appears with a $\pm 1$ coefficient.

  Conversely, it is easy to verify that if $\{t_a:a\in A\}\cup \{t_b:b\in B\} \cup \{t_c:c\in C\}$  does not satisfy the assertion, then either it is not possible to obtain in the determinant expansion the monomial  \eqref{monomial},  or  it is possible, but the resulting matrix after removing the corresponding $2s$ rows and columns is an $s\times s$ matrix with at least two identical columns, and therefore the monomial has a zero coefficient.

  The claims on the degree of $\det(M_{s, (I_1,I_2,I_3)})$ and the individual degree of each variable follow easily from the above.
\end{proof}
Given the set $S$ defined above, define
$$\mathbb{F}_S[x]=\{\det(M_{s,(I_1,I_2,I_3)}):\text{$(I_1,I_2,I_3)\in S$ and $s=\frac{1}{2}\sum_{1\le i<j\le 3}|I_i\cap I_j|$}\}$$
and note that
by Lemma \ref{determinant-1},
$\mathbb{F}_S[x]$ is a collection of at most $|S|$ \emph{nonzero} polynomials. The next lemma relates RS codes that are not  $(\frac{2}{3}(1-R),2)$ list-decodable to $\mathbb{F}_S[x]$.

\begin{lemma}\label{necessary}
  Let $\ma{C}$ be an $[n,k]$-RS code defined by an evaluation vector $\x{\alpha}=(\alpha_1,\ldots,\alpha_n)\in\mathbb{F}_q^n$.  If  $\ma{C}$  is \emph{not} $(\frac{2}{3}(1-R),2)$ list-decodable, then there exists a polynomial
  $f\in \mathbb{F}_S[x]$ such that $f(\x{\alpha})=0$.
\end{lemma}

\begin{proof}
    Since $\ma{C}$  is \emph{not} $(\frac{2}{3}(1-R),2)$ list-decodable,  there exist a vector $\x{y}\in\mathbb{F}_q^n$  and three distinct codewords $\x{c}_1,\x{c}_2,\x{c}_3\in B_{\fr{2(n-k)}{3}}(\x{y})\cap \ma{C}$.
    For $1\le i\le 3$, let  $Y_i=I(\x{c}_i,\x{y})$ and note that $|Y_i|\ge\fr{n+2k}{3}$. By the inclusion-exclusion principle,
  \begin{equation}\label{2-LD-formula-2}
    \begin{aligned}
      \sum_{1\le i<j\le 3} |Y_i\cap Y_{j}|-|\cap_{i=1}^3 Y_i|=\sum_{i=1}^3 |Y_i|-|\cup_{i=1}^3 Y_i|\ge 2k,
    \end{aligned}
  \end{equation}
  where the inequality follows from the fact that $|\cup_{i=1}^3 Y_i|\leq n$.   Let $I=\cap_{i=1}^3 Y_i$, then $I=I(\x{c}_1,\x{c}_2,\x{c}_3,\x{y})\s I(\x{c}_1,\x{c}_2,\x{c_3})$ and
  (\ref{2-LD-formula-2}) can be equivalently   written as
  $$\sum_{1\le i<j\le 3}|(Y_i\cap Y_{j})\setminus I|\ge 2k-2|I|.$$
  It is easy to check that for $1\le i\le 3$ there exist $I_i\s Y_i\setminus I$ such that 
  $$\sum_{1\le i<j\le 3}|I_i\cap I_{j}|= 2k-2|I|\geq 2,$$
  where the inequality holds, since by  the minimum distance of $\ma{C}$ we have $|I|\leq k-1.$   Further, observe that
  $I$ and  $I_i\cap I_{j}$ are disjoint subsets contained in    $I(\x{c}_i,\x{c}_{j})$, then again by the minimum distance of $\ma{C}$,  
  \begin{equation}\label{RS-property}
    |I_i\cap I_{j}|\le |I(\x{c}_i,\x{c}_{j})|-|I|\le k-1-|I|<k-|I|.
  \end{equation}
Finally, it is clear from the  definition of $I_i$ that $I_1\cap I_2 \cap I_3=\emptyset$, and we therefore conclude that  $(I_1,I_2,I_3)\in S$.

  Let $f_i\in \mathbb{F}_q^{<k}[x]$ be the  polynomial that corresponds to  codeword   $\x{c}_i$ for each $i\in[3]$.   For  distinct $i,j\in[3]$, let $g_{i,j}=\frac{f_i-f_{j}}{p_I(x,\x{\alpha})}\in \mathbb{F}_q^{<k-|I|}[x]$, where $p_{I}(x,\vec{\alpha})=\prod_{i\in I}(x-\alpha_i).$
   Note that $p_I(x,\x{v})$ indeed divides $f_i-f_{j}$ since $\{\alpha_i:i\in I\}\s Z(f_1-f_2)$.
It is easy to verify the following two facts
\begin{itemize}
\item $g_{1,2}+g_{2,3}+g_{3,1}=0$,
\item $g_{i,j}(\alpha_l)=0$ for $l\in I_i\cap I_j.$
\end{itemize}
These  equalities can be written in a matrix form as
$$M_{k-|I|,(I_1,I_2,I_3)}(\alpha)\cdot (g_{1,2},g_{3,1},g_{2,3})^T=0.$$
Since $f_i$ are distinct, $g_{i,j}$ are all nonzero, which implies further that there is a nontrivial solution to the set of equations $M_{k-|I|,(I_1,I_2,I_3)}(\alpha)\cdot \x{x}^T=0$. Thus, it follows that  $\det(M_{k-|I|,(I_1,I_2,I_3)}(\vec{\alpha}))=0$.
Equivalently, $\vec{\alpha}$  annihilates  the  polynomial $\det(M_{k-|I|,(I_1,I_2,I_3)})\in\mathbb{F}_{S}[x]$, which is nonzero by Lemma \ref{determinant-1} , as needed.
\end{proof}

\subsection{Proofs of Theorems \ref{2 list-optimal} and \ref{2 list-explicit}}

\no With Lemma \ref{necessary} at hand we can easily prove Theorems \ref{2 list-optimal} and \ref{2 list-explicit}.

\begin{proof}[\textbf{Proof of Theorem \ref{2 list-optimal}}]
It is clear that $|S|\leq 2^{3n}$, and that each polynomial of $\mathbb{F}_S[x]$ is of degree at most $k(k-1)$. Then, by Lemma \ref{zeros} and the union bound, the number of vectors  that annihilate a polynomial in $\mathbb{F}_S[x]$ is at most $2^{3n}k^2q^{n-1}.$ Furthermore, the number of vectors of $\mathbb{F}_q^n$ that have at least two identical coordinates  is at most $n^2q^{n-1}$. Hence, if $q> 2^{3n}k^2+n^2$, then there exist $(q-2^{3n}k^2-n^2)q^{n-1}$ vectors with $n$ distinct entries that do not annihilate any polynomial of $\mathbb{F}_S[x]$. By Lemma \ref{necessary} each such vector defines a $(\frac{2}{3}(1-R),2)$ list-decodable $[n,k]$-RS code, completing the proof of the theorem.\end{proof}

\begin{proof}[\textbf{Proof of Theorem \ref{2 list-explicit}}]
By construction we have  $[\mathbb{F}_q:\mathbb{F}_2]=k^n$ and $\mathbb{F}_q=\mathbb{F}_2(\alpha_1,\ldots,\alpha_n)$. Moreover, by (\ref{Eq-0}) it is not hard to verify that the $k^n$ field elements  $$\alpha_1^{i_1}\cdots\alpha_n^{i_n}:0\le i_1,\ldots,i_n\le k-1$$ form a basis for $\mathbb{F}_q$  over $\mathbb{F}_2$. In particular, they are linearly independent over $\mathbb{F}_2$, hence by Lemma \ref{determinant-1} $f(\alpha_1,\ldots,\alpha_n)\neq 0$ for any $f\in \mathbb{F}_S[x]$. Then, by Lemma \ref{necessary} we conclude that the $[n,k]$-RS code defined by $(\alpha_1,\ldots,\alpha_n)$ is indeed  $(\frac{2}{3}(1-R),2)$ list-decodable, as needed.
%
%
\end{proof}

\section{Cycle spaces, intersection matrices and their applications}\label{preparation}

\no In this section we give  the definition of the  cycle space of a graph which is then used to generalize the definition of $3$-wise intersection matrices to  $t$-wise intersection matrices, for any integer  $t\ge 3$. Then, we pose a conjecture (see Conjecture \ref{conjecture} below) on the nonsingularity of these matrices, and show that the correctness of this conjecture would imply also the correctness of  Conjecture \ref{conjecture-0}.

For  simplicity, in the sequel we assume that the alphabet size $q$ is a power of $2$.

\subsection{The cycle space of the complete graph}\label{cycle-space}

\no Let $K_t$ be the {\it complete graph} defined on vertices $[t]$. Let  $\{i,j\}$ and $ \Delta_{ijl}$  be the edge connecting vertices $i$ and  $j$, and the triangle with vertices $i,j,l$, respectively.
The set  $\binom{[t]}{2}$ (which is viewed as the set of edges of $K_t$) is ordered by the lexicographic order, as follows. For distinct $S_1,S_2 \in\binom{[t]}{2}$, $S_1< S_2$ if and only if $\max(S_1)<\max(S_2)$ or $\max(S_1)=\max(S_2)$ and $\min(S_1)<\min(S_2)$.

A {\it simple}\footnote{A graph is simple if it is undirected and contains no loops and multiple edges.} graph $G$ on vertices $[t]$ is uniquely represented by a \emph{binary} vector which by abuse of notation is {\it also} denote by  $G=(G_{ij}: \{i,j\}\in \binom{[t]}{2})\in\mathbb{F}_2^{\binom{t}{2}}$, where $G_{ij}=1$ if  and only if $\{i,j\}$  is an edge of $G$. Notice that the coordinates of the vector $G$ are ordered by the lexicographic order on  $\binom{[t]}{2}$.

A graph is said to be {\it Eulerian} if each of its vertices has even degree.
Let $C(G)$ denote the set of all Eulerian subgraphs of $G$; where in particular, $C(G)$ contains all cycles of $G$. Interestingly, $C(G)$ is closed under the operation of {\it symmetric difference} defined as follows.

Given $H,H'$, two Eulerian subgraphs of $G$, their symmetric difference $H\oplus H'$ is the subgraph that consists of all  edges that belong  to exactly \emph{one} of the graph $H$ or $H'$, but not to both of them. One can easily verify that the resulting graph is also Eulerian.
The symmetric difference  operation can also be described algebraically using the vector representation of the graphs $H, H'$. Indeed, the vector $H \oplus H'$, which is the sum over $\mathbb{F}_2$ of the vectors $H, H'\in\mathbb{F}_2^{\binom{t}{2}}$, represents the graph $H\oplus H'$. Hence, the set of vectors $H\in C(G)$ forms a linear subspace of $\mathbb{F}_2^{\binom{t}{2}}$, called  the {\it cycle space} of $G$. For our purposes we will only need the following result on the cycle space of  complete graphs.

\begin{lemma}[see, e.g., Theorem 1.9.5 of \cite{diestel2017graphentheory}]\label{cycle-space-basis}
The vector space $C(K_t)\subseteq \mathbb{F}_2^{\binom{t}{2}}$ has dimension $\binom{t-1}{2}$, and the set $\mathcal{B}_t=\{\Delta_{ijt}:1\le i<j\le t-1\}$ of $\binom{t-1}{2}$ triangles that share the common vertex $t$ forms a basis for it.
\end{lemma}

$\mathcal{B}_t$ is also viewed as a $\binom{t-1}{2}\times \binom{t}{2}$ binary matrix whose columns are labelled by the edges of $K_t$, according to the lexicographic order on $\binom{[t]}{2}$, and rows are the vectors $\Delta_{ijt}$ for $1\le i<j\le t-1$. For example, $\mathcal{B}_3=(1,1,1)$ and
\begin{equation}\label{Def-B_4}
  \begin{aligned}
    \mathcal{B}_4=\left(
      \begin{array}{cccccc}
        1 &  &  & 1 & 1 &  \\
         & 1 &  & 1 &  & 1 \\
         &  & 1 &  & 1 & 1 \\
      \end{array}
    \right),
  \end{aligned}
\end{equation}
where the $6$ columns are labeled and ordered lexicographically by  $\{1,2\}<\{1,3\}<\{2,3\}<\{1,4\}<\{2,4\}<\{3,4\}$, and  viewed as the edges of $K_4$. Observe for example that the $1$ entries in the first and second rows correspond to  basis vectors $\Delta_{124}$ and $\Delta_{134}$, respectively.

\subsection{$t$-wise intersection matrices}\label{int-matrix-def-gen}

\no Before we give the general definition of $t$-wise intersection matrices which is a bit involved, we explain the idea behind it. Consider an $[n,k]$-RS code defined over a field of characteristic $2$, and for $i\in [t]$ let $\vec{c}_i$ be $t$ distinct codewords of the code that correspond to polynomials $f_i$.  We view the polynomials $f_i$ as vertices of the complete graph $K_t$, and for an edge   $\{i,j\}$  incident to vertices  $f_i$ and $f_j$, let $f_{ij}:=f_i+f_j$ be the polynomial that corresponds to it, which is simply the difference between the vertices $f_i$ and $f_j$. These $\binom{t}{2}$ nonzero polynomials $f_{ij}$ (as the codewords $\x{c}_i$ are distinct so are the polynomials $f_i$) must satisfy many linear dependencies. Indeed, the summation of the polynomials $f_{ij}$ that correspond to the edges of any closed path in $K_t$ equals zero.

Putting this information in an algebraic language, define the  vector $\vec{f}=(f_{ij}: \{i,j\}\in \binom{[t]}{2})$ of length $k\binom{t}{2}$,  where we view each polynomial $f_{ij}$ as a vector of length $k$, and concatenate  the $\binom{t}{2}$ vectors (polynomials) according to the lexicographic order on $\binom{[t]}{2}$. Since for distinct $i,j,l$, $$f_{ij}+f_{il}+f_{jl}=0,$$ for a triangle $\Delta_{ijl}$ it follows that
$$(\Delta_{ijl}\otimes \mathcal{I}_k)\cdot \vec{f}^T=0$$

\noindent (note that $\Delta_{ijl}$ is viewed as a binary vector in $\mathbb{F}_2^{\binom{t}{2}}$). Since by Lemma \ref{cycle-space-basis} the triangles $\Delta_{ijl}$ span the cycle space $C(K_t)$, it follows that $(\x{u}\otimes \mathcal{I}_k)\cdot \vec{f}=0$ for any $\vec{u}\in C(K_t)$, or equivalently, $\vec{f}$ belongs to the kernel of the matrix $\mathcal{B}_t\otimes \mathcal{I}_k.$ A $t$-wise intersection matrix defined next is trying to capture this information together with some more information that follows form the fact that the code violates the required list-decodability (see Lemmas \ref{new-0} and \ref{claim-2} below).

\begin{definition}[$t$-wise intersection matrices]\label{Def-matrix-general}
For a positive $k$ and $t\geq 3$ subsets $I_1,\ldots,I_t\s[n]$, the $t$-wise intersection matrix  $M_{k,(I_1,\ldots,I_t)}$ is  the
$(\binom{t-1}{2}k+\sum_{1\le i<j\le t}|I_i\cap I_j|)\times \binom{t}{2}k$ variable matrix (in the variables $x_1,\ldots,x_n$) defined as

\begin{equation*}\label{int-matrix}
\left(
  \begin{array}{c}
    \mathcal{B}_t\otimes \mathcal{I}_k  \\\hline

  {\rm diag}(V_k(I_i\cap I_j): \{i,j\}\in \binom{[t]}{2}) \\
  \end{array}
\right),
\end{equation*}
where,
\begin{itemize}
\item $\mathcal{B}_t\otimes \mathcal{I}_k$ is a $\binom{t-1}{2}k\times \binom {t}{2}k$ binary matrix,

\item ${\rm diag}(V_k(I_i\cap I_j): \{i,j\}\in \binom{[t]}{2})$ is a block diagonal matrix with blocks  $V_k(I_i\cap I_j)$,  ordered by the lexicographic order on $\binom{[t]}{2}$. Note that this matrix has order $(\sum_{1\le i<j\le t}|I_i\cap I_j|)\times \binom {t}{2}k$, and if $I_i\cap I_j=\emptyset$ then $V_k(I_i\cap I_j)$ is of order $0\times k$. In other words, the $\{i,j\}\in \binom{[t]}{2}$ block of $k$ columns is a $\sum_{1\le i<j\le t}|I_i\cap I_j|\times k$ zero matrix.
\end{itemize}
\end{definition}

\begin{example}[4-wise intersection matrices]\label{Def-matrix-2}
Given four subsets $I_1,I_2,I_3,I_4\s [n]$, the 4-wise intersection matrix $M_{k,(I_1,I_2,I_3,I_4)}$ is the  $(3k+\sum_{
1\le i<j\le 4}|I_i\cap I_j|)\times 6k$ variable matrix

\begin{equation}\label{4-int-matrix-rep-1}
{\small\left(
  \begin{array}{cccccc}
    \mathcal{I}_k &      &     & \mathcal{I}_k   & \mathcal{I}_k &  \\
        &  \mathcal{I}_k &     & \mathcal{I}_k   &     & \mathcal{I}_k \\
        &      & \mathcal{I}_k &       & \mathcal{I}_k & \mathcal{I}_k \\\hline
    V_k(I_1\cap I_2) &  &  &  &  &  \\
     & V_k(I_1\cap I_3) &  &  &  &  \\
     &  & V_k(I_2\cap I_3) &  &  &  \\
     &  &  & V_k(I_1\cap I_4) &  &  \\
     &  &  &  & V_k(I_2\cap I_4) &  \\
     &  &  &  &  & V_k(I_3\cap I_4) \\
  \end{array}
\right).}
\end{equation}
\end{example}

Next, we give a justification to the definition of  $t$-wise intersection matrices in  Definition \ref{Def-matrix-general} by relating them to  list-decoding, as given in the following lemma.

\begin{lemma}\label{new-0}
  Let  $\x{c}_1,\ldots,\x{c}_t$ be $t$ distinct codewords of an $[n,k]$-RS code defined by  evaluation vector $\x{\alpha}$ over the field $\mathbb{F}_q$ of characteristic $2$. If $I_1,\ldots,I_t\s[n]$ are subsets satisfying $I_i\cap I_j\s I(\x{c}_i,\x{c}_j)$ for every $1\le i<j \le t$, then  there exists  a nonzero vector $\x{f}\in\mathbb{F}_q^{\binom{t}{2}k}$ such that
  \begin{equation}\label{Eq-10}
        M_{k,(I_1,\ldots,I_t)}(\x{\alpha})\cdot\x{f}^T=\x{0}.
  \end{equation}
\end{lemma}

\begin{proof}
  Define a vector $\x{f}=(f_{ij}:\{i,j\}\in \binom{[t]}{2})$ of length $k\binom{t}{2}$, where $f_{ij}=f_i+f_j$ and for  $1\le i\le t$, $f_i$ is the polynomial that corresponds to $\x{c}_i$.
  Similar to the proof of Lemma \ref{necessary}, it is easy to verify that  $\x{f}$ satisfies \eqref{Eq-10}. Furthermore, since the codewords $\x{c}_i$ are distinct, each  polynomial $f_{ij}$ is nonzero and in particular $\x{f}\neq 0$.
\end{proof}

\subsection{Nonsingular intersection matrices}\label{nonsingu}

\noindent Given subsets $I_1,\ldots,I_t\s[n]$, we call the variable matrix $M_{k,(I_1,\ldots,I_t)}$ {\it nonsingular} if it contains at least one $\binom{t}{2}k\times\binom{t}{2}k$ submatrix
whose determinant is a nonzero polynomial in $\mathbb{F}_q[x_1,\ldots,x_n]$.
In  this section we  pose a conjecture (see Conjecture \ref{conjecture}) that a certain condition  is sufficient for     $M_{k,(I_1,\ldots,I_t)}$ to be  nonsingular.
Then, we prove that if Conjecture \ref{conjecture} holds true, then so does  Conjecture \ref{conjecture-0}. To that end,  we define  next the weight of a collection of sets.

\begin{definition}[The weight of a collection of sets]\label{Def-weight-1}
The weight of $t$ sets $I_1,\ldots,I_t\s[n]$ is defined as
\begin{equation}
\label{chong-is-the-man}
{\rm wt}(I_1,\ldots,I_t):=\sum_{i=1}^t|I_i|-|\cup_{i=1}^tI_i|=\sum_{j=2}^t\sum_{J\s [t]:|J|=j}(-1)^{|J|}|\cap_{i\in J}I_i|,
\end{equation}
where the last equality follows from the inclusion-exclusion principle.  For a set $J$ of indices we denote also $\wt(I_J):=\wt(I_j:j\in J)$.    The  following useful inequality follows  easily from \eqref{chong-is-the-man}
\begin{equation}\label{inequality}
      \wt(I_1,\ldots,I_t)\ge\sum_{i=1}^t |I_i|-n.
  \end{equation}
\end{definition}


\begin{example}\label{example-wt}
  For subsets $I_1,I_2,I_3\s[n]$ it is easy to verify by \eqref{chong-is-the-man} that
  \begin{itemize}
    \item $\wt(I_1,I_2)=|I_1\cap I_2|$,
    \item $\wt(I_1,I_2,I_3)=|I_1\cap I_2|+|I_1\cap I_3|+|I_2\cap I_3|-|I_1\cap I_2\cap I_3|$.
  \end{itemize}
\end{example}
Now we are ready to pose the  conjecture.
\begin{conjecture}\label{conjecture}
  Let $t\ge 3$ be an integer and $I_1,\ldots,I_t\s[n]$ be subsets satisfying
  \begin{enumerate}
\item $\wt(I_J)\leq (|J|-1)k$ for any  nonempty $J\subseteq [t]$,
\item Equality holds for $J=[t]$, i.e., $\wt(I_{[t]})=(t-1)k$.
  \end{enumerate}
  Then, the $t$-wise intersection matrix $M_{k,(I_1,\ldots,I_t)}$ is nonsingular.
\end{conjecture}
The following theorem  relates  Conjecture \ref{conjecture} to Conjecture \ref{conjecture-0}.

\begin{theorem}\label{theorem-0}
  For a fixed integer $L\ge 2$, assuming  Conjecture \ref{conjecture} holds true for any integer $3\le t\le L+1$, then   Conjecture \ref{conjecture-0} holds true for $L$.
\end{theorem}

Unfortunately, we can only show that  Conjecture \ref{conjecture} holds for $t=3$ and $t=4$.
For $t=3$, the proof is an easy consequence of Lemma \ref{determinant-1} and is presented below.
However, the proof for  $t=4$ is much more involved, and is presented in Section \ref{proof-of-mian-theorem-2}.
In order to prove Theorem \ref{theorem-0}, we will need the following  lemma.

\begin{lemma}\label{claim-2}
If an $[n,k]$-RS code defined by the evaluation vector $\x{\alpha}$  is \emph{not} $(\fr{L}{L+1}(1-R),L)$ list-decodable for $L\geq 2$, then  there exists $t\in\{3,\ldots,L+1\}$ and subsets $I_1\ldots,I_t\s [n]$ such that
\begin{enumerate}
\item $\wt(I_J)\le(|J|-1)k$ for any nonempty   $J\subseteq [t]$,
\item $\wt(I_{[t]})=(t-1)k$.
\item $M_{k,(I_1,...,I_t)}(\x{\alpha})$ does not have full column rank.
\end{enumerate}
 \end{lemma}

\begin{proof}
Since the code is not $(\fr{L}{L+1}(1-R),L)$ list-decodable, there exist $L+1$ distinct codewords
$\x{c}_1,\ldots,\x{c}_{L+1}$ and a vector $\x{y}\in\mathbb{F}_q^n$ whose Hamming distance from any
$\x{c}_i$ is at most $\fr{L(n-k)}{L+1}$.
Let $Y_i=I(\x{c}_i,\x{y})$ for $1\le i\le L+1$. It is clear that  $|Y_i|\ge\fr{n+Lk}{L+1}$ and by (\ref{inequality}),
\begin{equation}
\label{tikitaka}
\wt(Y_1,\ldots,Y_{L+1})\ge\sum_{i=1}^{L+1} |Y_i|-n\ge Lk.
\end{equation}

Let $t$ be the smallest positive integer greater than $2$ for which there exist $t$ subsets $Y_{i_1},\ldots,Y_{i_t}$ with $\wt(Y_{i_1},\ldots,Y_{i_t})\geq (t-1)k$. By  \eqref{tikitaka}, $t$ is well defined and in particular $3\leq t\leq L+1$. Without loss of generality assume that $Y_1,\ldots,Y_t$ are the above $t$ sets.
Notice that the weight of a collection of sets can reduce by at most one if a single element is removed from exactly one of the sets.
 Furthermore, it is easy to verify that for any two distinct sets  $Y_i,Y_j$
$$\wt(Y_i)=0 \text{ and } \wt(Y_i, Y_j)\leq k-1,$$
where the last inequality follows from the easy fact that $Y_i\cap Y_j=I(\vec{c}_i,\vec{c}_j,\vec{y})\s I(\vec{c}_i,\vec{c}_j)$ and $|I(\x{c}_i,\x{c}_j)|\leq k-1$.
Hence, by the minimality of $t$ one can find subsets $I_i\s Y_i$ for $i\in [t]$ that satisfy the first two conditions. 
Moreover, by Lemma \ref{new-0} $M_{k,(I_1,\ldots,I_{t})}(\x{\alpha})$ does not have full column rank, as needed.
\end{proof}

\begin{proof}[\textbf{Proof of Theorem \ref{theorem-0}}]
We give an upper bound on the number of `bad' vectors $\x{\alpha}\in \mathbb{F}_q^n$ that do not define a  $(\fr{L}{L+1}(1-R),L)$ list-decodable $[n,k]$-RS code.
First, the number of vectors with at least two identical coordinates, and hence can not be an evaluation vector, is at most $\binom{n}{2}q^{n-1}$.
Second, by Lemma \ref{claim-2}, given a bad evaluation  vector $\x{\alpha}$ (whose entries are $n$ distinct field elements), there exist $t\in\{3,\ldots,L+1\}$ and subsets $I_1,\ldots,I_t$
that  satisfy the assumptions of Conjecture \ref{conjecture}, and $I_i\cap I_j\s I(\x{c}_i,\x{c}_j)$ for distinct $i,j$. By Lemma \ref{new-0}, the matrix  $M_{k,(I_1,\ldots,I_t)}(\x{\alpha})$ does not have full column rank. Assuming Conjecture \ref{conjecture} holds true, then $M_{k,(I_1,\ldots,I_t)}$ is nonsingular, i.e.,
it has a square submatrix of order $\binom{t}{2}k$, whose determinant is a nonzero polynomial in $\mathbb{F}_q[x_1,\ldots,x_n]$ which vanishes at $\x{\alpha}$.
It is clear that the number of such submatrices of $M_{k,(I_1,\ldots,I_t)}$ is at most $a_1(k,t)$ for some function of $k$ and $t$.
Furthermore, the degree of the determinant of  any such submatrix  is  bounded from above by another function $a_2(k,t)$.
Then, by  Lemma \ref{zeros} and the union bound, the number  of vectors that annihilate  at least one of these determinants is at most  $a_1(k,t)a_2(k,t)q^{n-1}.$
Summing over all possible values of $t$ and subsets $I_j$ we conclude that the total number of bad vectors is bounded from above by
$$\binom{n}{2}q^{n-1}+\sum_{t=3}^{L+1}\sum_{I_1,\ldots,I_t\s[n]} a_1(k,t)a_2(k,t)q^{n-1}\le \binom{n}{2}q^{n-1}+ \sum_{t=3}^{L+1} 2^{nt} a_1(k,t)a_2(k,t)q^{n-1}= c(n,k,L)q^{n-1},$$
completing the proof of Conjecture \ref{conjecture-0}.
\end{proof}



\begin{proposition}[The nonsingularity of 3-wise intersection matrices]
  Let $k$ be a positive integer and $I_1,I_2,I_3\s[n]$ be subsets satisfying
  \begin{enumerate}
    \item $\wt(I_J)\leq (|J|-1)k$ for any $J\subseteq [3]$,
    \item $\wt(I_{[3]})=2k$.
  \end{enumerate}
  Then, the 3-wise  intersection matrix $M_{k,(I_1,I_2,I_3)}$ is nonsingular, i.e., it contains a $3k\times 3k$ submatrix whose determinant is a nonzero polynomial.
\end{proposition}

\begin{proof}
If $I_1\cap I_2\cap I_3=\emptyset$, then the proposition follows by applying Lemma \ref{determinant-1} with $s=k$. On the other hand, if $I_1\cap I_2\cap I_3:=I\neq\emptyset$, then by setting $I'_i=I_i\setminus I$ for $i\in[3]$, it is not hard to check that
$$I'_1\cap I'_2\cap I'_3=\emptyset\text{\quad and\quad} |I'_i\cap I'_j|\le \frac{1}{2}\sum_{1\le i<j\le 3}|I'_i\cap I'_j|=2k-2|I|.$$

\noindent Consider the 3-wise intersection matrix $M_{k-|I|,(I'_1,I'_2,I'_3)}$. Clearly it is a square submatrix of $M_{k,(I_1,I_2,I_3)}$ of order $3k-3|I|$. Let $W$ be a submatrix of $M_{k,(I_1,I_2,I_3)}$, defined by its rows and columns containing $M_{k-|I|,(I'_1,I'_2,I'_3)}$ and the elements $\{x_i^{k-|I|},\ldots,x_i^{k-1}:i\in I\}$. It is clear that $W$ is a square submatrix of order $3k$, as each element $x_i^{j},i\in I,k-|I|\le j\le k-1$ is contained in precisely three rows and three columns of $M_{k,(I_1,I_2,I_3)}$. We claim that $\det(W)$ is a nonzero polynomial in $\mathbb{F}_q[x_1,\ldots,x_n]$. Indeed, in the determine expansion of $\det(W)$, the `coefficient' of the product
$$R(x):=\prod_{i\in I}\sum_{k-|I|\le j\le k-1} x_i^{3j}$$ is exactly  $\det(M_{k-|I|,(I'_1,I'_2,I'_3)})$.
Moreover, by applying Lemma \ref{determinant-1} with $s=k-|I|$ it follows that $\det(M_{k-|I|,(I'_1,I'_2,I'_3)})$ is a nonzero polynomial. Therefore $R(x)\cdot \det(M_{k-|I|,(I'_1,I'_2,I'_3)})$ appears in $\det(W)$ as a nonzero monomial, which implies that $\det(W)\neq 0$ and hence $M_{k,(I_1,I_2,I_3)}$ is nonsingular, as needed.
\end{proof}

\section{Optimal 3 list-decodable RS codes}\label{proof-of-mian-theorem-2}

\noindent  In this section we prove the following theorem which is Conjecture \ref{conjecture} for $t=4$.

\begin{theorem}[The nonsingularity of 4-wise intersection matrices]\label{3 list-main-lemma}
  Let $k$ be a positive integer and $I_1,I_2,I_3,I_4\s[n]$ be subsets satisfying
  \begin{enumerate}
    \item $\wt(I_J)\leq (|J|-1)k$ for any $J\subseteq [4]$,
    \item $\wt(I_{[4]})=3k$.
  \end{enumerate}
  Then, the 4-wise  intersection matrix $M_{k,(I_1,I_2,I_3,I_4)}$ is nonsingular, i.e., it contains a $6k\times 6k$ submatrix whose determinant is a nonzero polynomial.
\end{theorem}

The proof of Theorem \ref{3 list-optimal} follows easily from Theorem \ref{3 list-main-lemma} as follows.

\begin{proof}[\textbf{Proof of Theorem \ref{3 list-optimal}}]
By Lemma \ref{determinant-1} and Theorem \ref{3 list-main-lemma}, Conjecture \ref{conjecture} holds true for $t=3$ and $t=4$. Hence, by Theorem \ref{theorem-0}, Conjecture \ref{conjecture-0} holds true for $L=3$, which is exactly Theorem \ref{3 list-optimal}.
\end{proof}

We will prove Theorem \ref{3 list-main-lemma} by applying induction on $k$. Towards this end we will use a another representation of a 4-wise intersection matrix,  as stated below.
By appropriately permuting  the columns and  the first $3k$ rows of (\ref{4-int-matrix-rep-1}) one can  obtain the  matrix

\begin{equation}\label{4-int-matrix-rep-2}
\left(
  \begin{array}{c}
    \mathcal{I}_k\otimes \mathcal{B}_4  \\\hline

  (\ma{C}_i:0\le i\le k-1) \\
  \end{array}
\right)=
\left(
  \begin{array}{ccccc}
    \mathcal{B}_4 &  &  &  &  \\
     & \ddots &  &  &  \\
     &  & \mathcal{B}_4 &  &  \\
     &  &  & \ddots &  \\
     &  &  &  & \mathcal{B}_4 \\\hline
    \mathcal{C}_0 & \cdots & \mathcal{C}_i & \cdots & \mathcal{C}_{k-1} \\
  \end{array}
\right),
\end{equation}

\no where $\ma{C}_i={\rm diag}(V_k^{(i)}(I_j\cap I_l): \{j,l\}\in \binom{[4]}{2}$ and $V_k^{(i)}(I_j\cap I_l)$ is the $(i+1)$-th column of $V_k(I_j\cap I_l)$.
This representation turns out to be more amenable to applying induction.

\begin{example}
  For $k=2$, 
consider the following $4$-wise intersection matrix represented as in \eqref{4-int-matrix-rep-1}
  \begin{equation*}
  {\small\left(
  \begin{array}{c}
    \mathcal{B}_4\otimes \mathcal{I}_2  \\\hline

  {\rm diag}(V_k(I_i\cap I_j): \{i,j\}\in \binom{[4]}{2} \\
  \end{array}
\right)=
  \left(\begin{array}{cc|cc|cc|cc|cc|cc}
  1 &  &  &  &  &  & 1 &  & 1 &  &  &  \\
   & 1 &  &  &  &  &  & 1 &  & 1 &  &  \\\hline
   &  & 1 &  &  &  & 1 &  &  &  & 1 &  \\
   &  &  & 1 &  &  &  & 1 &  &  &  & 1 \\\hline
   &  &  &  & 1 &  &  &  & 1 &  & 1 &  \\
   &  &  &  &  & 1 &  &  &  & 1 &  & 1 \\\hline
  1 & x_1 &  &  &  &  &  &  &  &  &  &  \\
   &  & 1 & x_2 &  &  &  &  &  &  &  &  \\
   &  &  &  & 1 & x_3 &  &  &  &  &  &  \\
   &  &  &  &  &  & 1 & x_4 &  &  &  &  \\
   &  &  &  &  &  &  &  & 1 & x_5 &  &  \\
   &  &  &  &  &  &  &  &  &  & 1 & x_6
\end{array}\right).}
\end{equation*}
\no Then, the representation of the matrix as in \eqref{4-int-matrix-rep-2} takes the  form
\begin{equation}\label{hey}
{\small\left(
  \begin{array}{c}
    \mathcal{I}_2\otimes \mathcal{B}_4  \\\hline

  (\ma{C}_i:0\le i\le 1) \\
  \end{array}
\right)=
\left(\begin{array}{cccccc|cccccc}
  1 &  &  & 1 & 1 &  &  &  &  &  &  &  \\
   & 1 &  & 1 &  & 1 &  &  &  &  &  &  \\
   &  & 1 &  & 1 & 1 &  &  &  &  &  &  \\\hline
   &  &  &  &  &  & 1 &  &  & 1 & 1 &  \\
   &  &  &  &  &  &  & 1 &  & 1 &  & 1 \\
   &  &  &  &  &  &  &  & 1 &  & 1 & 1 \\\hline
  1 &  &  &  &  &  & x_1 &  &  &  &  &  \\
   & 1 &  &  &  &  &  & x_2 &  &  &  &  \\
   &  & 1 &  &  &  &  &  & x_3 &  &  &  \\
   &  &  & 1 &  &  &  &  &  & x_4 &  &  \\
   &  &  &  & 1 &  &  &  &  &  & x_5 &  \\
   &  &  &  &  & 1 &  &  &  &  &  & x_6
\end{array}\right).}
\end{equation}
\end{example}

The main idea behind the proof of Theorem \ref{3 list-main-lemma} is as follows.
Let $I_i,1\le i\le 4$ be  subsets satisfying the assumptions of the theorem, and let $M:=M_{k,(I_1,I_2,I_3,I_4)}$ be the corresponding 4-wise intersection matrix in the form of  (\ref{4-int-matrix-rep-2}).
To prove the nonsingularity of $M$, one needs to show that it has a square submatrix of order $6k$ whose determinant is a nonzero polynomial in $\mathbb{F}_q[x_1,\ldots,x_n]$. Observe that by (\ref{4-int-matrix-rep-2}), $M$ is composed of   $k$ block matrices of the form $\big(\fr{\ma{B}_4}{\ma{C}_{i}}\big)$\footnote{Here we simply ignore the zeros between $\ma{B}_4$ and $\ma{C}_{i}$.}, for  $0\le i\le k-1$. This structure is  utilized in the proof.
Loosely speaking, the proof   consists of the following two steps: 
\begin{itemize}
  \item [i.] We carefully pick three \emph{nonzero} elements of $\ma{C}_{k-1}$, where  no two of them are in the same row or column.
  Let $R$ be the $6\times 6$ submatrix of $\big(\fr{\ma{B}_4}{\ma{C}_{k-1}}\big)$ defined by the three rows of $\ma{B}_4$  and the three rows of $\ma{C}_{k-1}$ containing these three elements. Moreover,  let $r(x)$ be the product of the three elements, which is a nonzero polynomial in $\mathbb{F}_q[x_1,\ldots,x_n]$. Lastly,   let $M'$ be  the matrix obtained by  removing from $M$ the six rows and columns containing the submatrix $R$. We will show that $M'$ contains a submatrix  which is  a 4-wise intersection matrix $M_{k-1,(I'_1,I'_2,I'_3,I'_4)}$, where $I'_i\s I_i,1\le i\le 4$  are subsets satisfying the assumptions of Theorem \ref{3 list-main-lemma} for $k-1$.

  \item [ii.] By the induction hypothesis $M_{k-1,(I'_1,I'_2,I'_3,I'_4)}$  is nonsingular, hence  it contains a $6(k-1)\times 6(k-1)$  submatrix $W'$ whose determinant is a nonzero polynomial in $\mathbb{F}_q[x_1,\ldots,x_n]$.  Let $W$ be the $6k\times 6k$ submatrix of $M$ determined by the rows and columns of $W'$ and $R$. Suppose that we are able to show that $r(x)$ appears in $\det(R)$ with a nonzero coefficient (in particular $\det(R)\neq 0$), and the \emph{only} way to obtain $r(x)$  as a product of some elements of $W$ is by taking the  product of elements in $R$ (as described previously). Then it is easy to see that the `coefficient' of $r(x)$ in $\det(W)$ is $\beta\det(W')$ for a nonzero field element $\beta$ (in fact $\beta=1$). Therefore, the nonzero polynomial $r(x)\cdot\det(W')$ must appear in $\det(W)$, which implies that $\det(W)\neq 0$ and $M$ is nonsingular, as needed.
\end{itemize}

We remark that in Step (i),  the  three elements are  judiciously picked from $\mathcal{C}_{k-1}$. Indeed, since $r(x)$ should  appear as a nonzero monomial in $\det(R)$, the three columns of $R$ that do not contain the three picked elements, should be linearly independent once restricted to the rows of $\mathcal{B}_4$.
For example, let $R$ be  the submatrix  of $\big(\fr{\ma{B}_4}{\ma{C}_1}\big)$ in   (\ref{hey}) defined by the three elements $x_1,x_2,x_3$ of $\ma{C}_1={\rm diag}(x_i:1\le i\le 6)$.  Note that the last three columns of  $\ma{B}_4$
$$\left(
       \begin{array}{ccc}
         1 & 1 &  \\
         1 &  & 1 \\
          & 1 & 1 \\
       \end{array}
     \right)
$$
\no are linearly dependent over  fields of characteristic $2$, therefore $r(x):=x_1x_2x_3$ does not appear in the determinant expansion of  $R$.
On the other hand, if $R$ is defined by the elements $x_4, x_5, x_6$ then $r(x):=x_4x_5x_6$ does appear in $\det(R)$.

The above discussion motivates us to  characterize all sets of three columns of $\ma{B}_4$ that are linearly independent. The following easy to verify fact, provides such a characterization. 

\begin{fact}\label{fact-1}
  Three columns of $\mathcal{B}_4$ are linearly independent if and only if their labels (see the discussion below \eqref{Def-B_4}) are three edges which   \emph{do not} share a common vertex.
\end{fact}

As in $\ma{B}_4$, the columns of $\mathcal{C}_{i}$ and  $\big(\fr{\mathcal{B}_4}{\mathcal{C}_{i}}\big)$ are labeled by the edges $\{i,j\}\in\binom{[4]}{2}$  and ordered accordingly.
For notational convenience we define the {\it label} of an element in $\mathcal{B}_4$ or $\mathcal{C}_{k-1}$ to be the label of its column.
Then, we conclude that in step (i), the labels of the three picked elements  should not form a triangle, since otherwise the labels of the remaining  three columns of $\mathcal{B}_4$ will share a common vertex, making them linearly dependent by Fact \ref{fact-1}. To sum up, we arrive at the following conclusion.

\begin{fact}\label{fact-new-2} 
  Let $R$ be the $6\times 6$ submatrix defined by three nonzero elements of $\ma{C}_{k-1}$,   no two of them are in the same row or column. Then, by Fact \ref{fact-1},    the product of the three elements appears as a nonzero term in $\det(R)$, and in particular $\det(R)\neq 0$, if and only if their labels do not form a triangle.
\end{fact}

\subsection{Proof of Theorem \ref{3 list-main-lemma} for $k=1$}\label{proof}

\noindent We begin with some  needed notation.
Let the four subsets $I_i$ be given, then for a nonempty subset $J\subseteq[4]$ of indices, let $x_J$ be the number of elements that belong to each $I_j,j\in J$ and only to them, i.e.,
$$x_J=|\cap_{j\in J}I_j\backslash \cup_{j\notin J}I_j|.$$
The following fact follows easily  from  Definition \ref{Def-weight-1}.

\begin{fact}\label{fact-new-1}
  For a nonempty subset $J\subseteq [4]$, 
  \begin{equation}
  \label{ploplo}
      \wt(I_J)=\sum_{U: U\cap J\neq\emptyset}(|U\cap J|-1)x_U.
  \end{equation}
\end{fact}

\begin{proof}
By Definition \ref{Def-weight-1},  $\wt(I_J)=\sum_{j\in J}|I_j|-|\cup_{j\in J} I_j|$. Then, an element  $a\in \cup_{j\in J}I_j$ that  belongs to exactly $s$ of the $|J|$ subsets, contributes $s-1$ to $\sum_{j\in J}|I_j|-|\cup_{j\in J} I_j|$.
On the other hand, let $U$ be the set of indices of sets  that contain $a$. Then, $a$ contributes $1$ to the value of $x_U$ and in total $|U\cap J|-1=s-1$ to the right hand side of \eqref{ploplo}.
\end{proof}

\begin{example}\label{example-wt-2}
  The following equalities follow easily from Fact \ref{fact-new-1}
  \begin{align}
    &\wt(I_1,I_2)=x_{12}+x_{123}+x_{124}+x_{1234}\label{stam456}\\
    &\wt(I_1,I_2,I_3)=x_{12}+x_{13}+x_{23}+x_{124}+x_{134}+x_{234}+2x_{123}+2x_{1234}\label{stam457}\\
    &\wt(I_1,I_2,I_3,I_4)=\sum_{S\in\binom{[4]}{2}}x_S+2\sum_{T\in\binom{[4]}{3}}x_T+3x_{1234}\label{stam458}
  \end{align}
\end{example}

Notice that  $0\leq x_J\leq k$ for any $J\in\binom{[4]}{\ge2}$,  since otherwise by Fact \ref{fact-new-1} it follows that $\wt(I_J)\ge (|J|-1)x_J> (|J|-1)k$, and we arrive at a contradiction.

In the sequel,  capital $S$ (possibly with a subscript)  will always denote a $2$-subset of $[4]$  (or equivalently an edge of $K_4$). Similarly,  capital $T$ will denote a $3$-subset, i.e., a triangle of $K_4$. Lastly,  an arbitrary subset of $[4]$ will be denoted by $J$.  

Next we  present the proof of Theorem \ref{3 list-main-lemma} for the base case $k=1$.

\begin{proof}[{\it\textbf{Proof of Theorem \ref{3 list-main-lemma} for $k=1$}}]
   Let $I_i,1\le i\le 4$ be subsets satisfying the assumptions of the theorem for $k=1$, and let
   $M_{1,(I_1,I_2,I_3,I_4)}:=\big(\frac{\mathcal{B}_4}{\mathcal{C}_0}\big)$ be the corresponding 4-wise intersection matrix. Since $M_{1,(I_1,I_2,I_3,I_4)}$ is a binary matrix and not a variable matrix, it suffices to show that it contains an  invertible submatrix of order $6$. 
   Recall that  by assumption $\wt(I_J)\le |J|-1$ for any $J\subseteq [4]$, where   equality holds for $J=[4]$, and  $x_J=0, 1$ for    $J\in\binom{[4]}{\ge 2}$. It is easy to see that each row of $\ma{C}_0$ contains at most one 1 entry.
   The  proof is divided into three cases.

   \vspace{5pt}

   \noindent\textbf{Case 1. $x_{1234}= 1.$} Then, by \eqref{stam458}  $x_J=0$ for  $J\in\binom{[4]}{\ge 2}\setminus[4]$. 
   In this case, it is easy to verify that  $\ma{C}_0$ is the identity matrix of order six, as needed.

   \vspace{5pt}

   \noindent\textbf{Case 2. $x_{1234}=0$ and $x_T=0$ for every $T\in\binom{[4]}{3}$.}
   By \eqref{stam458}, $3=\wt(I_j:j\in[4])=\sum_{S\in\binom{[4]}{2}}x_S$. Since $x_S=0,1$ for any $S\in\binom{[4]}{2}$, there exist exactly three distinct subsets (edges) $S_1,S_2,S_3\in\binom{[4]}{2}$ with $x_{S_1}=x_{S_2}=x_{S_3}=1$. We claim that these  edges  do not form a triangle in $K_4$. Indeed, if they do form a triangle, say  $\Delta_{123}$, then by \eqref{stam457} $\wt(I_1,I_2,I_3)=x_{S_1}+x_{S_2}+x_{S_3}=3$, thereby violating assumption (1) of Theorem \ref{3 list-main-lemma}.
Since $x_{S_i}=1$ for $i=1,2,3$, column $S_i$   of $\ma{C}_0$   has at least one $1$ entry. For each $i$, pick a $1$ entry from column $S_i$, and let $R$ be the $6\times 6$ matrix defined by the rows of the  three  picked elements and the three rows of $\ma{B}_4$.
   Since the edges $S_1,S_2,S_3$ do not form a triangle, $R$ is invertible by Fact \ref{fact-new-2}, as needed.


   \vspace{5pt}

   \noindent\textbf{Case 3. $x_{1234}=0$ and there is at least one $T\in\binom{[4]}{3}$ with  $x_T= 1$.} By \eqref{stam458} it follows that there exists exactly one  $T\in\binom{[4]}{3}$ and  $S\in\binom{[4]}{2}$ with $x_T=x_S=1$.
 Moreover, note that  $S\nsubseteq T$, since otherwise by \eqref{stam457}, $\wt(I_T)\ge2x_T+x_S\ge 3$, which violates assumption (1) of Theorem \ref{3 list-main-lemma}.

    Let $S_1,S_2\in\binom{T}{2}$ be arbitrary two distinct edges of the triangle $T$. Since $x_T=1$, columns $S_1$ and $S_2$ of $\ma{C}_0$ contain  at least one $1$ entry. Similarly,
     column $S$ contains  at least one $1$ entry. Pick a $1$ entry from each of these three columns, and let $R$ be the $6\times 6$ matrix defined by the  rows of the  three picked elements and the three rows of $\ma{B}_4$.
   Since the edges $S,S_1,$ and $S_2$ do not form a triangle, $R$ is invertible by  Fact \ref{fact-new-2}, as needed.
   \end{proof}

\subsection{Proof of Theorem \ref{3 list-main-lemma} for $k\ge 2$}\label{proof-2}

\noindent Similar to the proof presented above, the proof of Theorem \ref{3 list-main-lemma} for $k\ge 2$ is also divided into different cases (in fact $5$ cases), however the proof for each case is much more technical. We give next the proof of the first two cases, and the proof for the remaining three cases is given in Appendix \ref{stasta}.

  \begin{proof}[\textbf{Proof of Theorem \ref{3 list-main-lemma} for $k\ge 2$}.]
  The proof of the theorem will follow by applying induction on $k$. Let $I_i,1\le i\le 4$ be subsets satisfying the assumptions of the theorem for $k\ge 2$, and let $M:=M_{k,(I_1,I_2,I_3,I_4)}$ be the corresponding 4-wise intersection matrix.
  \vspace{5pt}

  \noindent \textbf{Case 1. $x_{1234}\geq 1$.} Then, there exists  $a\in \cap_{i=1}^4 I_i$, and each column of  $\ma{C}_{k-1}$ contains the element  $x_a^{k-1}$.  Out of the six elements $x_a^{k-1}$ in $\ma{C}_{k-1}$, pick three  elements  whose labels  do not form a triangle, and let $R$  be the matrix defined by the rows containing the three picked elements and the rows of the matrix $\ma{B}_4$. Then, by Fact \ref{fact-new-2} $\det(R)\neq 0$ and   $r(x):=x_a^{3k-3}$, the product of the three picked elements, appears in $\det(R)$ as a nonzero term.
 Let $M'$ be the submatrix obtained by  removing from $M$ the rows and columns of $R$.    Set $I'_i:=I_i\setminus\{a\}$ for $1\le i\le 4$, and for a nonempty subset $J\in\binom{[4]}{\ge2}$ let $x'_J$ be defined similarly to $x_J$ but with the set  $I'_i,1\le i\le 4$ instead. Since $a$ was removed from each of the sets $I_i$, it easy to verify that   $x_J=x'_J$ for  $J\neq [4]$ and $x'_{1234}=x_{1234}-1$.
  Thus, it follows by \eqref{stam456}, \eqref{stam457}, \eqref{stam458} that for any nonempty  subset $J\in \binom{[4]}{\ge 2}$,
  $$\wt(I'_J)=\wt(I_J)-(|J|-1)\leq (|J|-1)(k-1),$$ where equality holds if $J=[4]$.
  Therefore, $I'_i,1\le i\le 4$ satisfy the assumptions of Theorem \ref{3 list-main-lemma} for $k-1$, and by the induction hypothesis   $M_{k-1,(I'_1,I'_2,I'_3,I'_4)}$ (which is a submatrix of $M'$)
  is nonsingular, i.e.,  it contains a $6(k-1)\times 6(k-1)$  submatrix $W'$ such that $\det(W')$ is a nonzero polynomial in $\mathbb{F}_q[x_1,\ldots,x_n]$.
  Let $W$ be the $6k\times 6k$ submatrix of $M$ defined by the rows and columns of  $W'$ and $R$, and notice that since neither $M_{k-1,(I'_1,I'_2,I'_3,I'_4)}$ nor $W'$ contain the variable $x_a$, the only way  to obtain the polynomial $r(x)$ as a product of some of the elements of $W$ is by taking the product of the three picked elements in the matrix $R$.  We claim that $\det(W)\neq 0$. Indeed, it is clear  that the `coefficient' of $r(x)$ in $\det(W)$ is exactly $\det(W')$. In other words,  the nonzero polynomial $r(x)\cdot\det(W')$ appears in $\det(W)$, which implies that $\det(W)\neq 0$ and $M$ is nonsingular, as needed.

  \vspace{5pt}

\noindent \textbf{Case 2. $x_{1234}=0$ and $x_T=0$ for every $T\in\binom{[4]}{3}$.} Then, by \eqref{stam458}   $3k=\sum_{S\in\binom{[4]}{2}} x_S$. We have the following claim.
  \begin{claim}\label{claim-new}
  There exists a vector $\vec{z}=(z_S:S\in\binom{[4]}{2})\in\{0,1\}^6$ such that for any $J\in \binom{[4]}{\ge 2}$
  $$\sum_{S\in\binom{J}{2}}z_S=\big\lfloor\frac{\wt(I_J)}{k}\big\rfloor \text{ or } \big\lceil\frac{\wt(I_J)}{k}\big\rceil.$$
  \end{claim}
Note that if $\wt(I_J)$ is divisible by $k$ then $\sum_{S\in\binom{J}{2}}z_S=\frac{\wt(I_J)}{k}$, and in particular for $J=[4]$ we have $\sum_{S\in\binom{[4]}{2}}z_S=3$. Hence, exactly three of the coordinates of $\vec{z}$ are $1$ and the remaining three coordinates are $0$.

Assuming the correctness of the claim, we proceed to show that $M$ is nonsingular.  Assume without loss of generality that  coordinates $z_{S_1}, z_{S_2}, z_{S_3}$ of $\vec{z}$ are $1$, and the remainings are $0$. We pick three nonzero elements of $\ma{C}_{k-1}$ according to the three nonzero coordinates of $\vec{z}$, as follows.

  Since for $i=1,2,3$, $1=z_{S_i}\leq \lceil\frac{\wt(I_{S_i})}{k}\rceil$, then
  $\wt(I_{S_i})=|\cap_{j\in S_i}I_j|>0$. Let $a_i\in \cap_{j\in S_i}I_j$ be an arbitrary element for $i=1,2,3$, and note that $a_1,a_2,a_3$ are distinct since we assumed that $x_{[4]}=x_T=0$ for any $T\in \binom{[4]}{3}$, i.e., each elements belongs to at most two of the sets $I_j,1\le j\le 4$. Next we claim that the sets (edges) $S_1,S_2,S_3$ do not form a triangle. Indeed, assume otherwise, i.e., $\cup_{i=1}^3 S_i=T\in\binom{[4]}{3}$, then by Claim \ref{claim-new}
  $$3=\sum_{S\in \binom{T}{2}}z_S\leq \big\lceil\frac{\wt(I_T)}{k}\big\rceil\leq \big\lceil \frac{2k}{k}\big\rceil,$$  a contradiction.

  Let $R$ be the $6\times 6$ submatrix of $\big(\fr{\ma{B}_4}{\ma{C}_{k-1}}\big)$ defined by the rows of $\ma{B}_4$ and the rows of $\ma{C}_{k-1}$ containing the elements  $x_{a_1}^{k-1},x_{a_2}^{k-1},x_{a_3}^{k-1}$. It follows from the above discussion and Fact \ref{fact-new-2} that $\det(R)\neq 0$, and in particular $\det(R)$ contains the nonzero monomial   $r(x):=x_{a_1}^{k-1}x_{a_2}^{k-1}x_{a_3}^{k-1}$. As before, let $M'$ be the matrix obtained by removing from $M$ the rows and columns of $R$.

  We claim that the subsets  $I'_i:=I_i\setminus\{a_1,a_2,a_3\}$ for  $1\le i\le 4$,  satisfy the assumptions of Theorem \ref{3 list-main-lemma} for $k-1$. Indeed, for  $J\in\binom{[4]}{\ge 2}$, let $x'_J$ be defined similarly to $x_J$ but with the sets  $I'_i,1\le i\le 4$.  It is easy to verify that
  \begin{align*}
  \wt(I'_{J})&= \sum_{S\in \binom{J}{2}}x'_S=\sum_{S\in \binom{J}{2}}(x_S-z_S)
  =\wt(I_{J})-\sum_{S\in \binom{J}{2}}z_S\\
  &\leq \wt(I_J)-\big\lfloor \frac{\wt(I_J)}{k}\big\rfloor\leq \begin{cases}
  3(k-1) &  J=[4]\\
  2(k-1) &  J\in \binom{[4]}{3}\\
  k-1 &  J\in \binom{[4]}{2},
  \end{cases}
  \end{align*}
and  that  equality holds for $J=[4]$.
 Therefore, by the induction hypothesis   $M_{k-1,(I'_1,I'_2,I'_3,I'_4)}$ (which is a submatrix of $M'$)
  is nonsingular, i.e.,  it contains a $6(k-1)\times 6(k-1)$  submatrix $W'$ such that $\det(W')$ is a nonzero polynomial in $\mathbb{F}_q[x_1,\ldots,x_n]$.
    Let $W$ be the $6k\times 6k$ submatrix of $M$ defined by the rows and columns of  $W'$ and $R$, and notice that since neither $M_{k-1,(I'_1,I'_2,I'_3,I'_4)}$ nor $W'$ contain the terms $x_{a_i}^{k-1}$ for  $i=1,2,3$, the only way  to obtain the polynomial $r(x)$ as a product of some of the elements of $W$ is by taking the product of the picked elements in the matrix $R$.  We claim that $\det(W)\neq 0$. Indeed, it is clear  that the `coefficient' of $r(x)$ in $\det(W)$ is exactly $\det(W')$. In other words,  the nonzero polynomial $r(x)\cdot\det(W')$ appears in $\det(W)$, which implies that $\det(W)\neq 0$ and $M$ is nonsingular, as needed.

 The proof of the remaining cases which completes the proof of the theorem is given in Appendix \ref{stasta}.
\end{proof}

 It remains to prove   Claim \ref{claim-new}. We remark that the proof is inspired  by the proof of Beck-Fiala Theorem  \cite{BECK19811} in Discrepancy theory.

 \begin{proof}[Proof of Claim \ref{claim-new} ]
 For the vector of variables $\x{y}=(y_S:S\in\binom{[4]}{2})$ and a subset $J\in\binom{[4]}{\ge 2}$ we define the  equation $E_J(\vec{y}):=\sum_{S\in \binom{J}{2}}y_S$, and say that the vector $\vec{z}\in \mathbb{R}^6$ {\it satisfies} it, if $E_J(\vec{z})$ is an integer, otherwise we say that the equation is {\it unsatisfied}.  Given a vector $\vec{z}\in \mathbb{R}^6$, let $\ma{S}$ and $\ma{U}$ be the set of satisfied and unsatisfied equations, respectively, then clearly  $|\ma{S}|+|\ma{U}|=11$, since  there are in total  $|\binom{[4]}{\ge 2}|=11$ equations.

 A natural candidate for the required vector $\vec{z}=(z_S:S\in\binom{[4]}{2})$, which we call the initial solution, is the vector  defined by $z_S=\fr{x_S}{k}$ for $S\in\binom{[4]}{2}$. Clearly, the initial solution  $\vec{z}$ satisfies for any  $J\in \binom{4}{\ge 2}$
  \begin{equation}
  \label{tiktak}
      \big\lfloor\frac{\wt(I_J)}{k}\big\rfloor\leq E_J(\vec{z}):=\sum_{S\in\binom{J}{2}}z_S\leq \big\lceil\frac{\wt(I_J)}{k}\big\rceil,
  \end{equation}
  and in particular for  $S\in \binom{[4]}{2}$
$$0\leq \big\lfloor\frac{\wt(I_S)}{k}\big\rfloor\leq E_S(\vec{z}):=z_S\leq \big\lceil\frac{\wt(I_S)}{k}\big\rceil\leq 1.$$
However,  we are looking for a vector $\vec{z}$ which  attains either the lower or upper bound in \eqref{tiktak} for each of the $11$ equations.
 Since $\wt(I_{[4]})=3k$, then by construction, the initial solution $\vec{z}$ satisfies equation $E_{[4]}$, and hence $\ma{S}$ is nonempty.  Moreover, the set of equations   `live' in the $6$ dimensional vector space of all linear  equations in variables $y_S,S\in \binom{[4]}{2}$ over $\mathbb{R}$.

Consider the set $\ma{S}$ of satisfied equations by the initial solution $\vec{z}$, and let $\rank(\ma{S})$ be the dimension of the vector space spanned by them. If $\rank(\ma{S})<6$, then the following algorithm modifies the initial solution to a solution such that its  set of satisfied equations (also denoted by $\ma{S}$) strictly contains  the previously satisfied equations  but also  $\rank(\ma{S})=6$.  Furthermore, throughout the execution of the algorithm,  \eqref{tiktak} would still hold for any $J\in\binom{[4]}{\geq 2}$, and therefore would also be satisfied by the modified solution  which is the output of the algorithm.

\vspace{5pt}
\IncMargin{1em}
\begin{algorithm}[H]
\SetKwInOut{Input}{input}\SetKwInOut{Output}{output}
\SetAlgoLined
\Input{$\vec{z}$ and sets $\ma{S},\ma{U}$}
 \While{$\rank(\ma{S})< 6$}{
  Pick an arbitrary  nonzero  $\vec{\beta}\in \mathbb{R}^6$ with $E(\vec{\beta})=0$ for each equation $E\in \ma{S}$\;
  Set $\vec{z}:=\vec{z}+\epsilon\vec{\beta}$ and $\epsilon=0$\;
  Gradually increase $\epsilon$ until $\vec{z}$ satisfies at least one  equation of $\ma{U}$\;
  Set $\ma{S}:=\ma{S}\cup \ma{S'}, \ma{U}:=\ma{U}\backslash \ma{S'}$, where $\ma{S'}$ is the set of satisfied equations in $\ma{U}$ by $\vec{z}$\;
  }
 \Output{$\vec{z}$}
\end{algorithm}
 \vspace{0.1in}

Let $\vec{z}$ be the output of the algorithm, and notice that its  set $\ma{S}$ of satisfied equations  contains  $6$  linearly independent equations, one of which is the equation $E_{[4]}$. By representing these $6$ equations in  a $6\times 6$ matrix $A$, it follows that $A\cdot \vec{z}$ is an integer vector.
We have the following claim whose proof is given in Appendix  \ref{appendix-A}.

\begin{claim}\label{determinant-A}
$\det(A)=\pm1$.
\end{claim}

Then, since $A$ is an integer matrix  with $\det(A)=\pm1$, it follows that  $A^{-1}$ is also an integer matrix, and then $\vec{z}$ is an integer vector. Lastly, since throughout the algorithm execution, $0\leq z_S\leq 1$ for any $S\in \binom{[4]}{2}$, it follows that  $z_S=0$ or $1$, as needed. Next, since $E_J \text{ for } J\in \binom{[4]}{\ge 2}$ is an equation with integer coefficients and $\vec{z}$ is an integer vector,  $E_J(\vec{z})$ is an integer satisfying \eqref{tiktak}. Therefore, $E_J(\vec{z})=\big\lfloor\frac{\wt(I_J)}{k}\big\rfloor \text{ or } \big\lceil\frac{\wt(I_J)}{k}\big\rceil$,
as needed.
\end{proof}

\section{Conclusions and open questions}
\label{conclustions}
\no In this paper we showed that in terms of  list-decodability, random RS codes behave very much like random codes. More precisely, most RS codes over a large enough field are optimal $L$ list-decodable for $L=2,3$. In particular, we showed that RS codes with large rate  are combinatorially list-decodable beyond the so called Johnson radius, thus answering in the affirmative to one of the main open questions regarding list-decoding of RS codes. Note that this result was already shown by \cite{Rudra-Wootters}, however our approach is totally different, and it enables a much better control on the order of the list size. Furthermore,  we gave the \emph{first} explicit construction of such codes, i.e., RS codes which are combinatorially list-decodable beyond the Johnson radius.
Lastly, we complemented the existence results with a matching upper bound on the decoding radius, which is a natural generalization of the Singleton bound to list-decoding.

In the course of proving the results of this paper, we introduced several new notions that were not used before in the context of list-decoding.
We believe that they can be used to prove Conjecture \ref{conjecture-0} for arbitrary $L\ge 2$, and as a consequence show that RS codes achieve the so called list-decoding capacity, i.e.,  for any $R,\epsilon>0$, there exist RS codes with rate $R$ over a large enough finite field, that are combinatorial list-decodable from radius $1-R-\epsilon$ and list size at most
 $\frac{1-R-\epsilon}{\epsilon}.$

This work opens up many interesting questions for further research. We list some of them below.

\begin{enumerate}
\item In standard parameters of RS codes, the field size scales linearly with the length of the code $n$, therefore a list-decoding algorithm that runs in polynomial time in $n$, also runs in polynomial time in the size of the input measured in bits.
However, this is not the case  when the field size is much larger than the length of the code, which is exactly the scenario in   Theorems \ref{2 list-optimal} and \ref{2 list-explicit}.  This  nonstandard case may lead to weird behaviours, such as  trivial list-decoding algorithms. Indeed,  Johan Rosenkilde \cite{johan} suggested the following list-decoding algorithm to the code given in Theorem \ref{2 list-explicit}. Given a received codeword $\vec{y}=(y_1,\ldots,y_n)$, run over all $I\in \binom{[n]}{k}$ and find the unique polynomial $f\in \mathbb{F}_q^{<k}[x_1,\ldots,x_n]$ with $f(\alpha_i)=y_i$ for all $i\in I$. Then, output $f$ if the Hamming distance between the codeword that corresponds to the polynomial $f$ and $\vec{y}$ is at most $\frac{2(n-k)}{3}$. Clearly, this algorithm  is basically Information-set decoding applied for list-decoding.  Since each symbol of the field is $k^n$ bits long,  the algorithm in fact runs in sublinear time in the input size.
To avoid this, the `right' algorithmic question is whether these code are list-decodable in polynomial time in the length of the code, where the complexity of the algorithm only measures the number of field operations needed, regardless of size of the field.

\item One possible way to approach the previous question is by  utilizing the information on the evaluation points. Indeed, recall that  Guruswami-Sudan algorithm \cite{gurus} can efficiently list-decode \emph{any} RS code up to the Johnson radius, no matter what are the specific  evaluation points of the code. In other words, the algorithms is oblivious to the evaluation points, and will always succeed  to decode up to the Johnson radius. On the other hand Theorem \ref{2 list-explicit} provides an explicit RS code with very specific evaluation points that is combinatorial list-decodable beyond the Johnson radius for codes with rate at least $1/4$. This raises the question whether it is possible to utilize the information on the evaluation points to decode beyond the Johnson radius?

\item Theorem \ref{2 list-optimal} and Theorem \ref{2 list-explicit}  show the abundance of RS codes that are combinatorially  list-decodable beyond the Johnson bound, however the required field is exponential (or double exponential in the case of the explicit construction) in the length of the code. Is it possible to reduce the field size, or one can show that such a behaviour of the field size is essential? This question is in fact related to the well-known MDS conjecture which relates the length of a code that attains the Singleton bound, to the alphabet size. Since  Theorem \ref{singleton-type} is a generalized Singleton bound, one can ask whether there exists a corresponding MDS conjecture, and if so, what is it?
    \item Prove Conjectures \ref{conjecture} and \ref{conjecture-0}.
\item Given an evaluation vector $(\alpha_1,\ldots,\alpha_n)$, find an efficient algorithm that verifies  whether the corresponding RS code attains the generalized Singleton bound for $L=2,3$.
\item Does a code that attain the generalized Singleton bound for some $L$ attains it also for any $L'<L$? This question is already interesting for $L=2$. Note that RS codes codes attaining the Singleton bound for $L=2$ already attain it for $L'=1$ as they are MDS code. Does this behaviour hold in general?
\item The generalized Singleton bound in \eqref{mish}  is known to be tight for $L=1$. Is it also tight for  $L>1$? 
\item Restricting to linear codes, by Theorems \ref{2 list-optimal} and \ref{3 list-optimal} it is known that the bound \eqref{mish2} is tight for $q>L$ and $L=2,3$. Is it tight for other parameters?
\end{enumerate}

{\small
\bibliographystyle{plain}
\bibliography{list-decoding}
}
\appendix
\section{Proof of Claim \ref{determinant-A}}
\label{appendix-A}
\noindent The claim can be easily verified by a computer and is given next for completeness.
Let $A$ be a  $6\times 6$ matrix whose rows represent  $6$ satisfied linearly independent equations, one of which is the equation $E_{[4]}$. Let $\ma{E}\cup \ma{T}$ be the remaining $5$ equations, where $\ma{E}$ is the set of satisfied edge equations, i.e., $E_S$ for $S\in \binom{[4]}{2}$, and $\ma{T}$ is the set of satisfied  triangle equations, i.e.,  $E_T$ for $T\in \binom{[4]}{3}$.  Clearly  $|\ma{E}|+|\ma{T}|=5$, and  $|\ma{T}|\le 3$, since the $5$ equations $E_{[4]}$ and $E_T,T\in \binom{[4]}{3}$  are linearly dependent.
The proof  is divided to several cases  according to the size of $|\ma{T}|$, as follows.
 \begin{itemize}
  \item  $|\ma{T}|=0$. Then,  by performing row operations on $A$, one can obtain  a permutation matrix whose determinant is $\pm1$,  as needed.

  \item $|\ma{T}|=1$ and  $\ma{T}=\{E_T\}$ for some $T\in \binom{[4]}{3}$. We claim that $\ma{E}$ contains exactly two  of the three edge equations $E_S, S\in\binom{T}{2}$. Indeed,  on one hand, $E_T$ and the $3$ equations $E_S, S\in \binom{T}{2}$ are linearly dependent. On the other hand, if it contains at most one of them, then since $|\ma{E}|=4$   it must contain all equations $E_S, S\in \binom{[4]}{2}\backslash \binom{T}{2}$. But, then these $3$ equations together with equations $E_{[4]}$ and $E_T$ are linearly dependent, and we arrive at a contradiction.
 Hence, let  $E_{S_1},E_{S_2}$ be the only equations  in $\ma{E}$ with    $S_1,S_2\in \binom{T}{2}$. Then, by performing row operations one can obtain the equation $E_T-E_{S_1}-E_{S_2}$ (which is an edge equation) instead of equation $E_T$,  and we return to the case $|\ma{T}|=0$.

  \item $|\ma{T}|=2$ and $\ma{T}=\{E_{T_1},E_{T_2}\}$. If $\ma{E}$ contains two  edge equations $E_{S_1},E_{S_2}$ with $S_1,S_2\in \binom{T_1}{2}$, then  similar to the previous case,  by performing row operations one can obtain the edge equation $E_{T_1}-E_{S_1}-E_{S_2}$  instead of equation $E_{T_1}$,  and we return to the case $|\ma{T}|=1$. Hence,  we can assume that for $i=1,2$, $\ma{E}$ contains at most one edge equation $E_{S_i}, S_i \in \binom{T_i}{2}$. Moreover, since
  $|\ma{E}|=3$, $\ma{E}$ must contain $E_S$ where $S$ is the only edge that is not a member of either of the triangles $T_1,T_2$, and   $E_{T_1\cap T_2}\notin \ma{E}$. We conclude that $\ma{E}=\{E_S,E_{S_1}, E_{S_2}\}$, where
  $S_i\in \binom{T_i}{2}, S_i\neq T_1\cap T_2$ for $i=1,2$. It is easy to verify that the edges $S,S_1,S_2$ form a path of length $3$ which is possibly closed, i.e., it might form a triangle. By performing row operations on  the two possible $A$'s (whether  the path is closed or not), one can obtain a permutation matrix, as needed.
  \phantom{each triangle by performing row operations one can obtain the equation    $|\ma{E}|=3$, then by symmetry, assume without loss of generality that $\ma{T}=\{\vec{1}_{[3]},\vec{1}_{\{1,2,4\}}\}$. By a reasoning similar to the previous case, we can always assume that $|\ma{E}\cap\{\vec{1}_{\{1,2\}},\vec{1}_{\{1,3\}},\vec{1}_{\{2,3\}}\}|\le 1$ and $|\ma{E}\cap\{\vec{1}_{\{1,2\}},\vec{1}_{\{1,4\}},\vec{1}_{\{2,4\}}\}|\le 1$. Since $|\ma{E}|=3$, it is easy to check by the pigeonhole principle that $\vec{1}_{\{3,4\}}\in\ma{E}$ and $\vec{1}_{\{1,2\}}\not\in\ma{E}$. It follows that $\ma{E}$ contains exactly one of $\vec{1}_{\{1,4\}},\vec{1}_{\{2,4\}}$ and exactly one of $\vec{1}_{\{1,3\}},\vec{1}_{\{2,3\}}$. In all of the four situations it is not hard to verify (by symmetry, we only need to verify two situations) that by performing certain row eliminations $A$ can be transformed to a permutation matrix, as needed.}

  \item $|\ma{T}|=3$. Let $\ma{T}=\{E_{T_1},E_{T_2},E_{T_3}\}$ and $\ma{E}=\{E_{S_1},E_{S_2}\}$. We claim that the edges $S_1, S_2$ are not disjoint, i.e., $S_1\cap S_2\neq \emptyset$. Indeed, assume otherwise, then since $\ma{T}$ is missing only one triangle equation and each edge is contained in two triangles, at least one of the edges, say $S_1$, is contained in two of the three triangles, say $T_1,T_2$. Further, by assumption, $S_2$ is disjoint to $S_1$, and therefore $S_2$ is not an edge of either of the triangles $T_1,T_2$. Then, it is easy to verify that the $5$ equations $E_{[4]},E_{S_1} ,E_{S_2}, E_{T_1}$, and $E_{T_2}$ are linearly dependent, and we arrive at a contradiction.

  Next we claim that each of the edges $S_1, S_2$ belongs to exactly one of the triangles $T_1,T_2,T_3$. Indeed, assume otherwise, say $S_1$ belongs to two of the triangles, then since $S_1,S_2$ form a path of length $2$, it is easy to verify that $S_1,S_2$ belong to one of the $T_i$'s. Hence, similarly to the previous case, one can reduce this case to the case of $|\ma{T}|=2$.
  We conclude that the each $S_i$ belong to exactly one of $T_i$'s, and then by  performing row operations on $A$, one can obtain a permutation matrix, as needed.
  \phantom{can reduce
  first of all we can show that $\ma{E}$ cannot consist of two disjoint edges. For example, if $\ma{E}=\{\vec{1}_{\{1,2\}},\vec{1}_{\{3,4\}}\}$, then it is not hard to check that $\vec{1}_{[4]},\vec{1}_{\{1,2\}},\vec{1}_{\{3,4\}},\vec{1}_{\{1,2,4\}},\vec{1}_{[3]}$ are linearly dependent, and so are $\vec{1}_{[4]},\vec{1}_{\{1,2\}},\vec{1}_{\{3,4\}},\vec{1}_{\{1,3,4\}},\vec{1}_{\{2,3,4\}}$. It thus follows that any choice of three triangle constraints will result in a linearly dependent $\ma{G}$, and we arrive at a contradiction. Next, we assume without loss of generality that $\ma{T}=\{\vec{1}_{[3]},\vec{1}_{\{1,2,4\}},\vec{1}_{\{1,3,4\}}\}$. As before, we can also assume that for each of the triangles listed above, $\ma{E}$ contains at most one of its edges. Then, it follows that $|\ma{E}\cap\{\vec{1}_{\{2,3\}},\vec{1}_{\{2,4\}},\vec{1}_{\{3,4\}}\}|\ge 1$. Again, by symmetry let us assume that $\vec{1}_{\{3,4\}}\in\ma{E}$. Hence, according to the discussion above, it is easy to see that $\vec{1}_{\{1,3\}}\not\in\ma{E},\vec{1}_{\{1,4\}}\not\in\ma{E}$ (since otherwise $\ma{E}$ will contain at least two edges of the triangle $\{1,3,4\}$) and $\vec{1}_{\{1,2\}}\not\in\ma{E}$ (since $\{1,2\}$ is disjoint from $\{3,4\}$). It follows that either $\ma{E}=\{\vec{1}_{\{2,3\}},\vec{1}_{\{3,4\}}\}$ or $\ma{E}=\{\vec{1}_{\{2,4\}},\vec{1}_{\{3,4\}}\}$. Moreover, those two situations are symmetric and it is not hard to verify that in either case by performing certain row eliminations $A$ can be transformed to a permutation matrix, as needed.}
 \end{itemize}

\section{Completing the proof of Theorem \ref{3 list-main-lemma} }
\label{stasta}
\noindent \textbf{Case 3. $x_{1234}=0$ and there exists exactly one subset $T\in\binom{[4]}{3}$ with $x_{T}\ge 1$.} Let $S_1,S_2,S_3\in \binom{[4]}{2}\backslash \binom{T}{2}$ be the three edges that do not belong to the  triangle $T$. For $i=1,2,3$, let $T_{\overline{S_i}}\in \binom{[4]}{3}$ be the only triangle distinct from $T$ that does not contain $S_i$ as an edge.

The proof relies on the  following lemma which shows that there exists $j\in [3]$ such that

\begin{equation}
\label{tamir2}
0<\wt(I_{S_j}), \wt(I_{T_{\overline{S_j}}})<2k \text{ and }\wt(I_{S_i})<k \text{ for } i\in [3]\backslash \{j\}.
\end{equation}

\begin{lemma}
\label{tamir}
One of the following holds
\begin{itemize}
    \item $\wt(I_{S_i})=k$ for exactly one $i=1,2,3$, and for this $i$ it also holds that $\wt(I_{T_{\overline{S_i}}})<2k$
    \item $\wt(I_{S_i})<k$ for  $i=1,2,3$, but for one of the $i$'s   $0<\wt(I_{S_i})$ and $\wt(I_{T_{\overline{S_i}}})<2k$.

\end{itemize}
\end{lemma}
\begin{proof}
 First assume that $\wt(I_{S_1})=k$, then
  by \eqref{stam457}
  $$2k\geq \wt(I_{S_1\cup S_2})\ge x_{S_1}+x_{S_2}+x_{T}= k+\wt(I_{S_2})+ x_T,$$
  hence $\wt(I_{S_2})<k$, since $x_T>0$. The proof for $\wt(I_{S_3})$ is analogous.
Next,   by \eqref{stam457} and  \eqref{stam458}
  $$3k=\wt(I_{[4]})\ge x_{S_1}+\wt(I_{T_{\overline{S_1}}})+x_{T},$$
  hence $\wt(I_{T_{\overline{S_1}}})<2k$, and the first part of the lemma follows.

 Next, assume that  $\wt(I_{S_i})<k$ for $i=1,2,3$. We claim that for at least two $i$'s, $0<\wt(I_{S_i})$, and for at least two $j$'s, $\wt(I_{T_{\overline{S_j}}})<2k$, and hence by the pigeonhole principle there exists an $1\leq a \leq 3$ for which both hold, i.e., $0<\wt(I_{S_a})$ and $\wt(I_{T_{\overline{S_a}}})<2k$.

For the first part, assume for contradiction that $\wt(I_{S_2})=\wt(I_{S_3})=0$, then by \eqref{stam457} and \eqref{stam458}
$$3k=\wt(I_{[4]})=x_{S_1}+x_{S_2}+x_{S_3}+\wt(I_T)\leq x_{S_1}+2k,$$ which implies that $x_{S_1}=\wt(I_{S_1})=k$, a contradiction to our assumption.

For the second part, assume for contradiction that $\wt(I_{T_{\overline{S_1}}})=\wt(I_{T_{\overline{S_2}}})=2k$, and  note that
$S_3$ belong to both triangles $T_{\overline{S_1}},T_{\overline{S_2}}$, then again by \eqref{stam457} and \eqref{stam458}
$$3k=\wt(I_{[4]})\ge \wt(I_{T_{\overline{S_1}}})+\wt(I_{T_{\overline{S_2}}})-\wt(I_{S_3}),$$
  which implies that $\wt(I_{S_3})=k$,
 a contradiction to our  assumption.
\end{proof}
   Since $x_{T}\geq 1$, there exists     $a\in \cap_{i\in T}I_i$.  Without loss of generality assume that by Lemma \ref{tamir}, \eqref{tamir2} holds with $j=1$, then  since $0<\wt(I_{S_1})=x_{S_1}$ there exists an element $b$ that belongs to sets $I_j,j\in S_1$ and only to them. Next,   let $S',S''\in \binom{T}{2}$ be arbitrary two edges of the triangle $T$, and  let   $R$ be the $6\times 6$ submatrix of $\big(\fr{\ma{B}_4}{\ma{C}_{k-1}}\big)$ defined by the three rows of $\ma{B}_4$ and the three rows of $\ma{C}_{k-1}$ defined by the two elements $x_a^{k-1}$ in columns  $S',S''$  and the element  $x_b^{k-1}$ in column $S_1$. As the labels of the three picked elements from $\ma{C}_{k-1}$ do not form a triangle, it follows by Fact \ref{fact-new-2} that $\det(R)\neq 0$, and in particular $r(x):=x_a^{2k-2}x_b^{k-1}$  appears in $\det(R)$ as a nonzero term.

As before, let $M'$ be the matrix obtained by removing from $M$ the rows and columns of $R$ and set $I'_i:=I_i\setminus\{a,b\}$ for $1\le i\le 4$. We claim that subsets $I'_i,1\le i\le 4$ satisfy the assumptions of Theorem \ref{3 list-main-lemma} for $k-1$. For any $J\in\binom{[4]}{\ge 2}$, let $x'_J$ be defined similarly to $x_J$ but with the sets $I'_i,1\le i\le 4$. Then, it is clear that $x'_{S_1}=x_{S_1}-1,x'_{T}=x_{T}-1$ and $x'_J=x_J$ for all other $J\in\binom{[4]}{\ge 2}\setminus\{S_1,T\}$.
 Then, it is not hard to check by \eqref{stam457} and \eqref{stam458} that
 \begin{itemize}
   \item $\wt(I'_{[4]})=\wt(I_{[4]})-3=3k-3$,
   \item for any $S\in\{\binom{T}{2}, S_1\}$, $\wt(I'_S)=\wt(I_S)-1\le k-1$,
   \item for any $T'\in\binom{[4]}{3}\setminus T_{\overline{S_1}}$, $\wt(I'_{T'})=\wt(I_{T'})-2\le 2k-2$,
   \item $\wt(I'_{T_{\overline{S_1}}})= \wt(I_{T_{\overline{S_1}}})-1\le 2k-2.$
 \end{itemize}
 Moreover, by Lemma \ref{tamir}  we have $\wt(I'_{S_2})\le\wt(I_{S_2})\le k-1$ and $\wt(I'_{S_3})\le\wt(I_{S_3})\le k-1$. Consequently, we conclude that $I'_i,1\le i\le 4$ satisfy the assumptions of Theorem \ref{3 list-main-lemma} for $k-1$, as claimed.
 Then, by the induction hypothesis the matrix $M_{k-1,(I'_1,I'_2,I'_3,I'_4)}$ (which is a submatrix of $M'$) is nonsingular, i.e., it contains a $6(k-1)\times 6(k-1)$  submatrix $W'$ such that $\det(W')$ is a nonzero polynomial in  $\mathbb{F}_q[x_1,\ldots,x_n]$. Let $W$ be the $6k\times 6k$ submatrix of $M$ defined by the rows and columns of $W'$ and $R$. We claim that $\det(W)\neq 0$. Indeed, observe that in $W$, the term $x_a^{k-1}$ appears exactly twice and the term $x_b^{k-1}$ appears exactly once. Hence, the only one way to obtain $r(x)$ as a product of elements of $W$ is by taking  the product of the three picked elements in the matrix $R$. Hence, the `coefficient' of $r(x)$ in $\det(W)$ is $\det(W')$ and $M$ is nonsingular.
  \vspace{5pt}

 \noindent \textbf{Case 4. $x_{1234}=0$ and there exist two subsets $T_i\in\binom{[4]}{3}$ with $x_{T_i}\ge 1$ for $i\in\{1,2\}$.}  Assume without loss of generality that $x_{123},x_{124}>0$  and $x_{134}=x_{234}=0$. The rest of the proof is divided to two subcases, according to the value of  $\wt(I_3,I_4)$.

  \vspace{5pt}
  \noindent\textbf{Subcase 4.1.} $\wt(I_3,I_4)=x_{34}=k$. The proof  follows by  reducing this case to Case $3.$ by showing that \eqref{tamir2} holds.

  We claim that $\wt(I_1,I_4),\wt(I_2,I_4)<k$ and  $\wt(I_1,I_2,I_4)<2k$. Since  $x_{123}>0$ and $x_{134}=x_{234}=0$, then $\wt(I_1,I_4)=x_{14}+x_{124}$, and  by \eqref{stam457}
 $$2k\ge\wt(I_1,I_3,I_4)\ge x_{34}+x_{123}+\wt(I_1,I_4)> k+\wt(I_1,I_4),$$
 \noindent which implies that $\wt(I_1,I_4)< k$, as needed. The proof that $\wt(I_2,I_4)<k$ is analogous.
 Next,  by \eqref{stam457} $$\wt(I_1,I_2,I_4)=x_{12}+x_{14}+x_{24}+x_{123}+2x_{124},$$ then  by \eqref{stam458}
 $$3k=\wt(I_{[4]})\ge x_{34}+x_{123}+\wt(I_1,I_2,I_4)> k+\wt(I_1,I_2,I_4),$$
 \noindent which implies that $\wt(I_1,I_2,I_4)<2k$.

  One can easily verify that \eqref{tamir2} holds with $T=\{1,2,3\}, S_1=\{3,4\},S_2=\{2,4\}, S_3=\{1,4\}$,  $T_{\overline{S_1}}=\{1,2,4\}$, and $j=1$. Hence, the proof follows by Case 3.

  \vspace{5pt}
  \noindent\textbf{Subcase 4.2.} $\wt(I_3,I_4)<k$. Let  $a\in I_1\cap I_2\cap I_3$ and $b\in I_1\cap I_2\cap I_4$, and pick from $\ma{C}_{k-1}$  the two elements  $x_a^{k-1}$ in columns $\{1,3\}$, and $\{2,3\}$ and the element  $x_b^{k-1}$ is column  $\{1,4\}$.
 Let $R$ be the $6\times 6$ submatrix of $\big(\fr{\ma{B}_4}{\ma{C}_{k-1}}\big)$ defined by the three rows of $\mathcal{B}_4$ and the three rows of the three picked elements from $\ma{C}_{k-1}$.
 As the labels of the elements picked from $\ma{C}_{k-1}$ do not form a triangle, it follows by Fact \ref{fact-new-2} that $\det(R)\neq 0$, and in particular,  the product of the three picked elements $r(x):=x_a^{2k-2}x_b^{k-1}$ appears in $\det(R)$ as a nonzero term.

Set $I'_i:=I_i\setminus\{a\}$ for $1\le i\le 3$ and $I'_4:=I_4\setminus\{b\}$. We claim that subsets $I'_i,1\le i\le 4$ satisfy the assumptions of Theorem \ref{3 list-main-lemma} for $k-1$. For any $J\in\binom{[4]}{\ge 2}$, let $x'_J$ be defined similarly to $x_J$ but with the sets $I'_i,1\le i\le 4$. Then,
 \begin{itemize}
   \item $x'_{123}=x_{123}-1,x'_{124}=x_{124}-1,x'_{12}=x_{12}+1$,
   \item for $J\in\binom{[4]}{\ge2}\setminus\{\{1,2,3\},\{1,2,4\},\{1,2\}\}$, $x'_J=x_J$.
 \end{itemize}
 Indeed, the first two equalities hold since $I'_1\cap I'_2\cap I'_3=I_1\cap I_2\cap I_3\setminus\{a\}$ and $I'_1\cap I'_2\cap I'_4=I_1\cap I_2\cap I_4\setminus\{b\}$; the third equality holds since by definition $b$ appears three times in $I_i,1\le i\le 4$, but two times in $I'_i,1\le i\le 4$; all other equalities are easy to verify. Observe that $x_{12}\le k-2$, since $x_{123}\ge 1,x_{124}\ge 1$ and $k\ge \wt{I_1,_2}= x_{12}+x_{123}+x_{124}$. Therefore, it is not hard to check by Fact \ref{fact-new-1} that for any $J\in\binom{[4]}{\ge2}\setminus\{3,4\}$,
 $$\wt(I'_J)\le\wt(I_j:j\in J)-(|J|-1)\le(|J|-1)(k-1),$$
 where  equality holds for $J=[4]$. Moreover, by assumption $\wt(I'_3,I'_4)\leq \wt(I_3,I_4)< k$, therefore   $I'_i,1\le i\le 4$ satisfy the assumptions of Theorem \ref{3 list-main-lemma} for $k-1$, as claimed. 

 As before, by the induction hypothesis the matrix $M':=M_{k-1,(I'_1,I'_2,I'_3,I'_4)}$ is nonsigular, i.e., it contains a $6(k-1)\times 6(k-1)$  submatrix $W'$ such that $\det(W')$ is a nonzero polynomial in $\mathbb{F}_q[x_1,\ldots,x_n]$. Let $W$ be the $6k\times 6k$ submatrix of $M$ defined by the rows and columns of $W'$ and $R$. We claim that $\det(W)\neq 0$. Observe that in $W$, $x_a^{k-1}$ appears exactly twice with labels $\{1,3\}$ and $\{2,3\}$, and $x_b^{k-1}$ also appears exactly twice  with labels  $\{1,2\}$ and $\{1,4\}$. Therefore, in the determinant expansion of $W$ we have at most two possible ways to obtain $r(x)=x_a^{2k-2}x_b^{k-1}$ as a product of elements of $W$, namely, we have two choices for the term $x_b^{k-1}$, either in column $\{1,2\}$  (which is denoted by $x_{b,\{1,2\}}^{k-1}$) or the one in column $\{1,4\}$  (which is denoted by $x_{b,\{1,4\}}^{k-1}$). Let us consider the term $x_a^{2k-2}x_{b,\{1,2\}}^{k-1}$. Since the labels of the two $x_a^{k-1}$'s and $x_{b,\{1,2\}}^{k-1}$ form a triangle, then by  Fact   \ref{fact-new-2} the term  $x_a^{2k-2}x_{b,\{1,2\}}^{k-1}$ vanishes in the determinant expansion of $W$. We conclude that the  only one way to obtain $r(x)$ as a product of elements of $W$ is by the term $x_{b,\{1,4\}}^{k-1}$, which appears as a nonzero term in   $\det(R)$.  The final step of the proof is similar to that of the previous cases, hence it is omitted.

  \vspace{5pt}

  \noindent\textbf{Case 5. $x_{1234}=0$ and there exist at least three subsets $T_i\in\binom{[4]}{3}$ with $x_{T_i}\ge1$ for $i\in\{1,2,3\}$}. The proof will follow by reducing this case to Subcase $4.2.$ Recall, that the proof followed by the existence of two triangles $T_1,T_2\in \binom{[4]}{3}$ and the edge $S\in\binom{[4]}{2}$, which is the only edge contained in neither of the triangles $T_i$, such that
  $x_{T_1},x_{T_2}>0 \text{ and } \wt(I_S)<k$. We show next that this also holds true in this case.

  Assume without loss of generality that $x_{123},x_{124},x_{134}>0$ and  that $\wt(I_3,I_4)$ attains the minimum among $\wt(I_2,I_3),\wt(I_2,I_4),\wt(I_3,I_4)$. We claim that $\wt(I_3,I_4)<k$.
  Towards a contradiction,  assume the $\wt(I_2,I_3)=\wt(I_2,I_4)=\wt(I_3,I_4)=k$. Hence,  by Example \ref{example-wt}
  $$2k\ge\wt(I_2,I_3,I_4)=|I_2\cap I_3|+|I_2\cap I_4|+|I_3\cap I_4|-|I_2\cap I_3\cap I_4|=3k-|I_2\cap I_3\cap I_4|,$$
  which implies that $|I_2\cap I_3\cap I_4|=k$. Since $I_2\cap I_3\cap I_4\subseteq I_2\cap I_3$, and $|I_2\cap I_3|\leq k$ it follows that  $I_2\cap I_3\cap I_4=I_2\cap I_3$. However, by assumption we have $I_1\cap I_2\cap I_3\neq\emptyset$, then  $$\emptyset\neq I_1\cap I_2\cap I_3=I_1\cap (I_2\cap I_3\cap I_4),$$ violating the assumption that $x_{1234}=0$.
We conclude that $\wt(I_3,I_4)<k, x_{123},x_{124}>0$ and the proof follows from subcase $4.2.$ as claimed.
\end{document}